\documentclass[12pt,reqno]{amsart}
\usepackage{amsfonts,amsmath}
\usepackage{amssymb}
\usepackage{textcomp}
\usepackage{enumitem}
\usepackage{graphicx,color}
\numberwithin{equation}{section}
\newcommand{\nc}{\newcommand}
\nc{\parent}[1]{$[\![#1]\!]$}
\newtheorem{theorem}{Theorem}[section]
\newtheorem{lemma}{Lemma}[section]
\newtheorem{example}{Example}[section]

\newtheorem{remark}{Remark}[section]
\newtheorem{definition}{Definition}[section]
\newtheorem{assumption}{Assumption}[section]

\newenvironment{pf-main}{{\sc Proof of Theorem \ref{mainresult}.}\hspace{3mm}}{\qed}
\DeclareMathOperator{\sgn}{sgn}

\nc{\cadlag}{c\`{a}dl\`{a}g } \nc{\caglad}{c\`{a}gl\`{a}d }
\nc{\ba}{\begin{array}} \nc{\ea}{\end{array}}
\nc{\be}{\begin{equation}} \nc{\ee}{\end{equation}}
\nc{\bea}{\begin{eqnarray}} \nc{\eea}{\end{eqnarray}}
\nc{\bean}{\begin{eqnarray*}} \nc{\eean}{\end{eqnarray*}} \nc{\wtilde}{\widetilde}
\nc{\bu}{\bullet} \nc{\nn}{\nonumber} \nc{\cA}{{\mathcal A}}
\nc{\cB}{{\mathcal B}} \nc{\cC}{{\mathcal C}}  \nc{\cD}{{\mathcal
		D}}\nc{\sD}{{\mathscr
		D}} \nc{\sE}{{\mathscr
		E}} \nc{\sB}{{\mathscr
		B}}\nc{\cL}{{\mathcal L}} \nc{\cN}{{\mathcal
		N}}\nc{\bbD}{\mathbb{D}} \nc{\cG}{{\mathcal G}} \nc{\cE}{{\mathcal E}}\nc{\cT}{{\mathcal T}} \nc{\cF}{{\mathcal
		F}} \nc{\cS}{{\mathcal S}} \nc{\cR}{{\mathcal
		R}}\nc{\cU}{{\mathcal U}} \nc{\cH}{{\mathcal H}} \nc{\bfQ}{\mathbf{Q}}
\nc{\cK}{{\mathcal K}} \nc{\cM}{{\mathcal M}} \nc{\cP}{{\mathcal
		P}} \nc{\bbE}{\mathbb{E}} \nc{\bbH}{\mathbb{H}} \nc{\bbF}{\mathbb{F}} \nc{\bbG}{\mathbb{G}}\nc{\bbEQ}{\mathbb{E}_{\mathbb{Q}}}
\nc{\eps}{\varepsilon} \nc{\bfE}{\mathbf{E}}\nc{\bbN}{\mathbb{N}}
\nc{\bbEP}{\mathbb{E}_{\mathbb{P}}}\nc{\bbL}{\mathbb{L}}
\nc{\bbP}{\mathbb{P}} \nc{\bbQ}{\mathbb{Q}} \nc{\Om}{\Omega} \nc{\bbW}{\mathbb{W}}
\nc{\om}{\omega} \nc{\bbR}{\mathbb{R}} \nc{\bbC}{\mathbb{C}}
\nc{\bfr}{\begin{flushright}} \nc{\efr}{\end{flushright}}
\nc{\dXt}{\Delta X_{t}} \nc{\dXs}{\Delta X_{s}}
\nc{\bs}{\blacksquare} \nc{\dX}{\Delta X} \nc{\dY}{\Delta Y}
\nc{\dnkx}{\left(X(T^{n}_{k})-X(T^{n}_{k-1})\right)}
\nc{\dom}{depth-of-the-market } \nc{\uar}{\uparrow}
\nc{\dar}{\downarrow}\nc{\rar}{\rightarrow}
\nc{\half}{\frac{1}{2}}
 \nc{\hbE}{\hat{\bbE}}

\nc{\what}{\widehat} \nc{\fhat}{\what{f}}\nc {\parx}{\frac{\partial}{\partial x}} \nc{\parw}{\frac{\partial}{\partial w}} \nc{\parww}{\frac{\partial^2}{\partial w^2}} \nc{\bfT}{\mathbf{T}}
\def\rar{\rightarrow}
\def\dar{\downarrow}

\nc{\esssup}{\mathrm{ess}\mbox{ }\mathrm{sup}}
\nc{\chf}{\mbox{$\mathbf1$}}
\nc{\red}{\color{red}}
\begin{document}

\title{On pricing rules and optimal strategies in general Kyle-Back models}
\author[]{Umut \c{C}etin}
\address{Department of Statistics, London School of Economics and Political Science, 10 Houghton st, London, WC2A 2AE, UK}
\email{u.cetin@lse.ac.uk}
\author[]{Albina Danilova}
\address{Department of Mathematics, London School of Economics and Political Science, 10 Houghton st, London, WC2A 2AE, UK}
\email{a.danilova@lse.ac.uk}
\date{\today}
\begin{abstract}
The folk result in Kyle-Back models states that the value function of the insider remains unchanged when her admissible strategies are restricted to absolutely continuous ones.  In this paper we show that,  for a large class of pricing rules used in current literature, the value function of the insider can be finite when her strategies are restricted to be absolutely continuous and infinite when this restriction is not imposed. This implies that the folk result doesn't hold for  those pricing rules and that they are not consistent with equilibrium.  We derive the necessary conditions for a pricing rule to be consistent with equilibrium and prove that, when a pricing rule satisfies these necessary conditions, the insider's optimal strategy is absolutely continuous, thus obtaining the classical result in a more general setting.  

This, furthermore, allows us to justify the standard assumption of absolute continuity of insider's strategies since one can construct a pricing rule satisfying the necessary conditions derived in the paper  that yield the same price process as the pricing rules employed in the modern literature when insider's strategies are absolutely continuous.

\end{abstract}
\maketitle

\section{Introduction} 
The canonical model of markets with asymmetric information is due to
Kyle \cite{Kyle}, where he studies a market for a single risky asset whose price is determined in equilibrium. Kyle set up the  model  in discrete time and conjectured continuous-time equilibrium by considering the limit.  The continuous-time framework  was formalised by Back \cite{Back} and thus the model is commonly referred to as a Kyle-Back model in the subsequent literature. In this type of models there are typically three types of agents participating in the market: non-strategic noise traders, a strategic risk-neutral informed trader (insider) with private information regarding the future value of the
asset, and a number of risk-neutral market makers competing for the total demand.   The goal of market makers is to set the {\em pricing rule} so that the resulting price is rational, which in particular entails finite expected profit for the insider trading at these prices.  On the other hand, the objective of the insider is to maximise her expected final wealth given the pricing rule set by the market makers. Thus, this type of modelling can be viewed as a game with asymmetric information between the market makers and the insider and the goal is to find an equilibrium of this game.

Apart from extending the Kyle's model to continuous time the most important contribution of Back \cite{Back} was to establish that, when the market maker sets the price to be a harmonic function of total order, the insider's value function is finite and the optimal control solving the insider's optimisation problem is absolutely continuous. This implies that the set of admissible controls of the insider can be reduced to absolutely continuous ones. This restriction significantly simplifies the problem of finding an equilibrium since it allows one to employ a PDE approach to the insider's optimal control problem that yields a system of PDEs that the value function of the insider and the pricing rule of the market maker have to satisfy in equilibrium.

Back's result was the original justification for restricting the set of admissible controls of the insider to absolutely continuous ones and this restriction is now standard in the asymmetric information literature (see, e.g., \cite{Back-Baruch}, \cite{Cho}, \cite{CDF}, \cite{CCD}, \cite{CRH}, \cite{CFNO}, \cite{CNF}, and \cite{MSZ}). In this paper we show that if we extend the class of pricing rules beyond harmonic functions of total order to include ones used in the recent literature, e.g. in the papers cited above, then the value function of the insider is infinite and her optimal control is not absolutely continuous. In particular, this is true for the pricing rules in the aforementioned papers. Since the value function of the insider is infinite, those pricing rules can not be equilibrium pricing rules.

However, since the infinite profit is due to penalty imposed on discontinuous strategies or strategies with additional martingale part being insufficient to offset the profit made due to private information, one can modify this penalty to ensure optimality of absolutely continuous strategies, while warranting the same price process when insider's strategy is absolutely continuous.  This is precisely what we do in this paper by establishing a class of pricing rules that yield the same price process as the models cited before when the trading strategy of the insider is absolutely continuous but produce a finite value for the insider when her strategies are allowed to have jumps or martingale parts. We show that for this class of pricing rules the set of admissible controls of the insider can be reduced to absolutely continuous ones. 

To the best of our knowledge, this paper is the first one to identify this class of pricing rules consistent with an equilibrium. Moreover, it is also the first one since \cite{Back} that justifies the restriction of insider's controls to absolutely continuous ones in a general setting. Thus it closes the gap between the assumption of absolutely continuous controls and its justification in the modern literature that employs more general pricing rules.

The paper is structured as follows. In Section 2 we describe the model and introduce the set of pricing rules that generalise the pricing rules employed in the current literature. {In Section 4 we state the main results of the paper:  Theorem \ref{t:jmprice} that derives the necessary conditions on the pricing rule that ensure that the insider cannot achieve infinite profits by employing discontinuous strategies and/or strategies with a martingale part, Theorem \ref{t:gzero} that establishes a PDE condition on the pricing rule that is necessary for the existence of equilibrium, and Theorem \ref{mm:t:AC} which demonstrate that, under the conditions for the pricing rule derived Theorems  \ref{t:jmprice}  and \ref{t:gzero}, the restriction of admissible controls to absolutely continuous ones produces the same value function. In Section 4 provide a worked-out example that illustrates how the techniques developed in this paper can be applied to a particular model. In Section 5 we analyse the optimisation problem of the insider and establish a subset of these pricing rules that yield a finite value to this problem. Moreover, as a by-product we obtain the familiar sufficient conditions on the pricing rule and the trading strategy in order for the equilibrium to exist.}

\section{Model setup}
As in \cite{Back} we will assume that the trading will take place over the time interval $[0,1]$. Let $(\Omega , \cG , (\cG_t)_{t \in [0,1]} , \bbQ)$ be a filtered probability space satisfying the usual conditions, The time-1 value of the traded asset is given by $f(Z_1)$, which will become public knowledge at $t=1$ to all market participants, where $Z$ is a continuous and adapted process, and $f$ is a  measurable increasing function.

Three types of agents  trade in the market. They differ in their information sets and objectives as follows.

\begin{itemize}
	\item \textit{Noise/liquidity traders} trade for liquidity reasons, and
	their total demand at time $t$ is given by a standard $(\cG_t)$-Brownian motion $
	B$ independent of $Z$.
	\item \textit{Market makers} only  observe   the total demand
	\[
	Y=\theta+B, 
	\]
	where $\theta$ is the demand process of the informed trader. The admissibility condition imposed later on $\theta$ will entail in particular that $Y$ is a semimartingale. 
	
	They set the price of the risky asset via a {\em Bertrand competition} and clear the market. Similar to \cite{Back-Baruch} we assume that the market makers set the price as a
	function of weighted total order process at time $t$, i.e.{ we assume that the price process, $S$, is given by
        \begin{equation} \label{mm:e:rule_mm}
 	S_t = H\left(t, X_t\right), \qquad \forall t\in [0,1)
	\end{equation}}
	where $X$ is adapted to the filtration generated by $B$ and $Z$ and is the unique strong solution of a certain SDE whose coefficients and drivers are constructed by the market makers as made precise in Definition \ref{mm:d:prule}.  Moreover, a pricing rule  has to be admissible in the sense of Definition \ref{mm:d:prule}, which will entail $S$ being a semimartingale.  
	\item \textit{The informed investor} observes the price process $
	S_{t}=H\left(t, X_t\right)$ and her private signal $Z$.
	Since she is
	risk-neutral, her objective is to maximize the expected final
	wealth, i.e.
	\bea
	&&\sup_{\theta \in \mathcal{A}}E^{0,z}\left[ W_{1}^{\theta
	}\right], \mbox{ where} \label{ins_obj}
	\\
	W_{1}^{\theta
	}&=&
	(f(Z_{1})-S_{1-})\theta _{1-}+\int_{0}^{1-}\theta _{s-}dS_{s}. \label{mm:eq:insW}
	\eea
	In above $\mathcal{A}$ is the set of admissible trading strategies
	for the given pricing rule\footnote{Note that this implies  the insider's
		optimal trading strategy takes into account the \emph{feedback
			effect}, i.e. that prices react to her trading strategy.}, which will be defined in Definition \ref{mm:d:iadm}. Moreover, $E^{0,z}$ is the expectation with respect to $P^{0,z}$, which is the regular conditional distribution of  $(X_s, Z_s; s\leq 1)$ given $X_0=0$ and $Z_0=z$, which exists due to Theorem 44.3 in \cite{Bauer}.
	
	Thus, the insider maximises the
	expected value of her final wealth
	$W_{1}^{\theta }$, where the first term on the right hand side of equation (%
	\ref{ins_obj}) is the contribution to the final wealth due to a potential
	differential between  the market price and the fundamental value at the time of information
	release, and the second term is the contribution to the final wealth coming from
	the trading activity.
\end{itemize}

Given the above market structure, we can now precisely define the filtrations of the market makers and of the informed trader.  As we shall consider their right continuous augmentations, we first define the probability measures that will be used in the augmentation. 

First  define $\cF:=\sigma(B_t,Z_t; t \leq 1)$ and let $Q^{0,z}$ be the regular conditional distribution of $(B,Z)$ given $B_0=0$ and $Z_0=z$.  Observe that any $P^{0,z}$-null set is also $Q^{0,z}$-null in view of the assumption on $X$.   Due to the measurability of regular conditional distributions one can define  the probability measure $\bbP$ on $(\Om, \cF)$  by
\be \label{mm:d:bbP}
\bbP(E)=\int_{\bbR} Q^{0,z}(E) \bbQ(Z_0\in dz),
\ee
for any $E \in \cF$.  

While $Q^{0,z}$ is how the informed trader assign likelihood to the events generated by $B$ and $Z$, $\bbP$ is the probability distribution of the market makers who do not observe $Z_0$ exactly. Thus, the market makers' filtration, denoted by $\cF^M$, will be the right-continuous augmentation with the $\bbP$-null sets of the filtration generated by $Y$. In particular $\cF^M$ satisfies the usual conditions. 

On the other hand, since the informed trader knows the value of $Z_0$ perfectly, it is plausible to assume that her filtration is augmented with the $Q^{0,z}$-null sets. However, this will make the modelling cumbersome since the filtration will have an extra dependence on the value of $Z_0$ purely for technical reasons.  Another natural choice is to consider the null sets that belong to every $Q^{0,z}$, i.e. the sets that are elements of the following
\be \label{mm:e:nullI}
\cN^I:=\{E\subset \cF: Q^{0,z}(E)=0, \, \forall z\in \bbR\}.
\ee
These null sets will correspond to the  {\em a priori} beliefs that the informed trader has about the model before she is given the private information about $Z_0$ and, thus, can be used as a good benchmark for comparison. Therefore we assume that the informed trader's filtration, denoted by $\cF^I$, is the right continuous augmentation\footnote{See  Section 3 of \cite{GTM} for a recipe of the procedure.} of the filtration generated by $S$ and $Z$  with the sets of $\cN^I$. Similarly, we will denote by $\cF^{B,Z}$ the is the right continuous augmentation of the filtration generated by $B$ and $Z$  with the sets of $\cN^I$. Note that the resulting filtrations are {\em not} complete.

A rational
expectations equilibrium is a pair consisting of
an \emph{admissible} pricing rule and an \emph{admissible}
trading strategy such that: \textit{a)} given the pricing rule
the trading strategy is optimal, \textit{b)} given the trading
strategy,  the pricing rule is {\em rational} in the following sense:
\be \label{mm:d:mm_obj}
H(t,X_t)=S_t=\mathbb{E}\left[   f(Z_1)|\mathcal{F}_t^M\right],
\ee
where $\bbE$ corresponds to the expectation operator under $\bbP$.
To formalize
this definition of equilibrium, we first  define the sets of admissible
pricing rules and trading strategies.

\begin{definition}\label{mm:d:prule} An {\em admissible
		pricing rule} is any quadruple $(H,w,c,j)$ fulfilling the following
	conditions:
	\begin{enumerate}
		\item $w : [0,1]\times \bbR \to (0,\infty)$ is a function in $C^{1,2}([0,1] \times \bbR)$;
		\item \label{mm:d:noempty} Given a Brownian motion, $\beta$, on some filtered probability space, there exists a unique strong solution to 
		\[
		d\tilde{X}_t=w(t,\tilde{X}_t)d\beta_t, \qquad \tilde{X}_0=0.
		\]
		\item $H \in C^{1,2}([0,1) \times \bbR)$.
		\item $x \mapsto H (t,x)$ is strictly increasing for every $t\in [0,1)$;
		\item $c:[0,1]\times \bbR \to \bbR$ is locally Lipschitz;
		\item  $j:[0,1]\times \bbR \times\bbR \to \bbR$ is continuous and there exists $\Delta^*>0$ such that 
		\be \label{mm:e:jgrowth}
		|j(t,x,\Delta)|\leq \Gamma(t,x) |\Delta|, \qquad \mbox{ for } |\Delta| <\Delta^*,
		\ee
		where $\Gamma$ is locally bounded.
	\end{enumerate}
\end{definition}
\begin{definition} \label{mm:d:iadm}
	An $\cF^{B,Z}$-adapted \footnote{ See Remark \ref{mm:r:insfilt} for the explanation of the choice of filtration $\theta$ is adapted to.}  $\theta$ is said to be an  admissible trading
	strategy for a  given admissible pricing rule $(H,w,c,j)$  if the following conditions are stisfied.
	\begin{enumerate} \item $\theta$ is  a semimartingale\footnote{Note that due to the incompleteness of the stochastic basis we follow the notion of semimartingale from Jacod and Shiryaev \cite{JS} that only requires the right-continuity of filtrations.} with summable jumps on $(\Om, \cF, (\cF^{B,Z}_t), Q^{0,z})$ for each $z \in \bbR$.
		\item There exists a unique strong solution\footnote{Following Kurtz \cite{Kurtz07} $X$ is a strong solution of (\ref{mm:eq:signal_mm1}) if there exists a measurable mapping, $\varphi$, from a Polish space to a Polish space such that $X:=\varphi(Y)$ satisfies (\ref{mm:eq:signal_mm1}). In this case both Polish spaces are taken to be the space of right continuous functions  with left limits on $[0,1]$ equipped with Skorokhod topology. Note that $Y$ may jump only due to the discontinuities in $\theta$.} , $X$, to 
		\[
		\begin{split}
	           dX_t = &w (t,X_{t-}) dY^c_t+\Big(\frac{w_x(t,X_{t-})}{2}+c(t,X_{t-})\Big)w(t,X_{t-})(d[Y,Y]^c_t-dt)\\
	                   &+ K_w^{-1}(t,j(t,X_{t-},\Delta Y_t)+K_w(t,X_{t-})+\Delta Y_t)-X_{t-}
	        \end{split}
	        \]
	starting from $X_0=0$ over the time interval $[0,1]$ on $(\Om, \cF, (\cF^{B,Z}_t), \bbP)$, where $Y
		=B+\theta$, and 
	\begin{equation}
	K_w(t,x)=\int_0^x\frac{1}{w(t,y)}dy + \half \int_0^t w_x(s,0)ds. \label{mm:e:Kwdef}
        \end{equation}
		\item No doubling strategies are
		allowed, i.e. for all $z \in \bbR$
		\begin{equation}
		E^{0,z}\left[ \int_{0}^{1}H^{2}\left(t,X_t\right)dt\right] <\infty.
		\label{mm:e:theta_cond_2}
		\end{equation}
		The set of admissible trading strategies for the  given pricing rule $(H,w,c,j)$ is denoted with $\mathcal{A}(H,w,c,j)$. For the notational brevity, we will also denote by $\mathcal{A}(H,w):= \mathcal{A}(H,w,0,0)$.
	\end{enumerate}
\end{definition} 

 Now we are in a position to define  the price set by the market makers. Observe that, given  an admissible pricing rule  $(H,w,c,j)$ and a trading strategy $\theta\in \mathcal{A}(H,w,c,j)$, there exists (due to the admissibility of $\theta$) unique strong solution of
	\begin{equation}\label{mm:eq:signal_mm1}
	dX_t = w (t,X_{t-}) dY^c_t+dC_t+ J_t ,\quad X_0 =0,
	\end{equation}
	with $Y=B+\theta$, $K_w$ defined in (\ref{mm:e:Kwdef}) and
	\bea
	dC_t&=&\Big(\frac{w_x(t,X_{t-})}{2}+c(t,X_{t-})\Big)w(t,X_{t-})(d[Y,Y]^c_t-dt), \label{mm:e:Cdef} \\
	J_t&=&K_w^{-1}(t,j(t,X_{t-},\Delta Y_t)+K_w(t,X_{t-})+\Delta Y_t)-X_{t-}. \label{mm:e:Jdef} 
	\eea
	Thus the process $X$ is well-defined\footnote{ See Remark \ref{r:welldefined} for details.}, and therefore 
	$$
	   S_t:=H(t,X_t),
        $$ 
        is well-defined. Moreover, $S$ is adapted to $(\cF_t^M)$ since, following Kurtz \cite{Kurtz07}, $X$ is a functional of the paths of $Y$ only and $Y$ is adapted to $(\cF^M_t)$.  { These considerations justify the definition of price stated below.}
        \begin{definition}\label{mm:d:price} Given  an admissible pricing rule  $(H,w,c,j)$ and a trading strategy $\theta\in \mathcal{A}(H,w,c,j)$ the price set by market maker is given by 
        $$
            S_t:=H(t,X_t),
        $$
        where $X$ is unique strong solution of (\ref{mm:eq:signal_mm1})
        \end{definition}
        
       Before we discuss our admissibility conditions, we will give a few examples from the literature  that illustrate the flexibility of our definition of price:
       
    \begin{example} 
    
    {In Back  \cite{Back}  the trading takes place in continuous time and the insider observes the final asset value from the start. More precisely,  the final value of an asset in this model is $f(Z_1)$, where $f$ is an increasing function and $Z_1$  is  Gaussian random variable with mean $0$ and variance $1$. 
    
    The price is defined to be:
       \be\label{eq:Back}
         S_t=H(t,Y_t).
       \ee
       In our setting, it translates to the pricing rule $(H, 1,0,0)$. Observe that $K_{w}(t,x)=x$ and therefore $J_t=\Delta Y_t$ yielding representation of (\ref{mm:eq:signal_mm1}) in the form:
       $$
          dX_t=dY^c_t+\Delta Y_t=dY_t,
       $$
       that is, $X_t=Y_t$ for all $t\in[0,1]$. In this setting, our definitions of admissibility of trading strategies and the pricing rules coincide with the ones in \cite{Back}, as well.
       
       Note that an alternative representation of the equilibrium stock price process is possible in this model. The following representation demonstrates that already in \cite{Back}  the market makers use different weights to price absolutely continuous, martingale and jump parts of the insider strategy.  
       
        To see this define $w(t,x):= H_y(t, H^{-1}(t,x))$, where the inverse is taken with respect to the space variable, $x$. Then the equilibrium pricing rule can be represented as $(\tilde{H},w,0,0)$, where $\tilde{H}(t,x)=x+H(0,0)$. Indeed, applying Ito formula for semimartingales to (\ref{eq:Back}) yields:
  
       $$
       dS_t=w(t,S_{t-})dY^c_t+\frac{1}{2}w_{x}(t,S_{t-})w(t,S_{t-})(d[Y]^c_t-dt)+H(t,Y_{t})-H(t,Y_{t-})
       $$ 
       due to the fact that the equilibrium pricing rule in \cite{Back}  satisfies
       $$
          H_t(t,y)+\frac{1}{2}H_{yy}(t,y)=0.
       $$
       This PDE also allows to conclude that $K_{w}(t,x)=H^{-1}(t,x)-H^{-1}(0,0)$, thus yielding that
       $$
          H(t,Y_{t})-H(t,Y_{t-})=K_w^{-1}(t, K_w(t,S_{t-})+\Delta Y_t)-S_{t-},
       $$
       which implies that the equilibrium pricing rule indeed can be represented as (\ref{mm:eq:signal_mm1}) with pricing rule $(\tilde{H},w,0,0)$ (and therefore $X_t=S_t-H(0,0)$).  This representation holds for any admissible trading strategy. This, in fact, provides a generalisation of equation (16) in \cite{Back} for general admissible pricing rules (and not only for equilibrium ones that are considered in Theorem 3). Observe that $H(t,Y_{t})-H(t,Y_{t-}) \neq H_y(t,Y_{t-})\Delta Y_t$. If one instead prices jumps via    $H_y(t,Y_{t-})\Delta Y_t$, the insider obtains infinite profits as we show in Theorem \ref{t:jmprice} complementing the illustration given by \cite{Back} in Footnote 6.} 
   \end{example}  
   
   {
   \begin{example} In Baruch \cite{Baruch} the trading takes place in continuous time where the risk-averse insider observes the normally distributed asset value from the start. That is, $f(Z_1)=aZ_1+b$ is Gaussian random variable with known mean and variance. Using the above notation $S_t=X_t$, where
\[
dX_t=\lambda(t)dY_t,
\]
and $\lambda$ is deterministic.

 In our setting, it translates to the pricing rule $(H, \lambda,0,0)$ with $H(t,x)=x$. Indeed, as $\lambda$ does not depends on $x$ and the evolution of $X$ in \cite{Baruch} does not depend on quadratic variation of the martingale part of $Y$ explicitly, we should have
 $$
    0=\left(\frac{\lambda_x(t,X_{t-})}{2}+c(t,X_{t-})\right)\lambda(t,X_{t-})=c(t,X_{t-})\lambda(t,X_{t-})
 $$
 yielding $c\equiv0$. Moreover, in this case $K_w(t,x)=\frac{x}{\lambda(t)}$ and therefore we should have
 $$
   \lambda(t)\Delta Y_t=\lambda(t)(j(t,X_{t-},\Delta Y_t)+\Delta Y_t)
 $$
 yielding $j\equiv 0$.
\end{example}

     \begin{example}\label{ex:BB} Back and Baruch \cite{Back-Baruch} study a version of the Kyle model where the trading takes place in continuous time; however, the public announcement for $f(Z_1)$ takes place at an exponential time $\tau$ with rate $r$. Moreover, $f$ is the identity mapping and $Z_1$ is a Bernoulli random variable taking values in $\{0,1\}$, and $\tau$ is independent from everything else. 
     
     Back and Baruch assume that $S_t=X_t$ and
\be\label{e:X_ex_3}
dX_t=\lambda(X_s)dY_s,
\ee
for some function $\lambda$.  In this model admissible trading strategies of the insider are absolutely continuous. If we allow trading strategies with martingale part and keep the same dynamics for $S$, the price postulated in  \cite{Back-Baruch}  will translate to $\left(H, \lambda,-\frac{\lambda_x}{2}, \cdot\right)$ with $H(t,x)=x$ in this setting.\footnote{As jumps in the insider's strategies are not allowed, the $j$ term is not defined here. To allow for jumps we need to define evolution for $X$ for a process $Y$ that is not continuous.} Indeed, as the evolution of $X$ in (\ref{e:X_ex_3}) does not depend on quadratic variation of the martingale part of $Y$ explicitly, we should have
 $$
    0=\left(\frac{\lambda_x(X_{t})}{2}+c(t,X_{t})\right)\lambda(X_{t})
 $$
 yielding the stated form of $c$.
     \end{example} 
     
     \begin{example}\label{ex:CCA}  In Campi, \c{C}etin and Danilova \cite{CCD} extension of the Kyle model the insider receives a dynamic signal given by a diffusion process
\[
dZ_t=\sigma(s)a(V(s),Z_s)dW_s,
\]
where $V$ is an increasing and absolutely continuous deterministic function and $a$ is a diffusion coefficient satisfying mild regularity assumptions. The market makers choose a pricing rule $H$ and set the market price $S_t=H(t,X_t)$, where
\[
dX_t= w(t,X_t)dY_t,
\]
where $w$ is a smooth weighting function. In this model it is assumed  that admissible trading strategies of the insider are absolutely continuous. If we generalize this model to allow trading strategies with jumps and martingale part by assuming that
\be\label{e:X_ex_4}
   dX_t= w(t,X_{t-})dY_t,
\ee
this pricing rule will translate to the setting of this paper as $\left(H,w,c,j\right)$ with $c(t,x)=-\frac{w_x(t,x)}{2}$ and 
$$
   j(t,x,\Delta)=K_w(t,x+w(t,x)\Delta)-K_w(t,x)-\Delta
$$
with $K_w$ defined as in (\ref{mm:e:Kwdef}). Indeed, as the evolution of $X$ in (\ref{e:X_ex_4}) does not depend on quadratic variation of the martingale part of $Y$ explicitly, we should have
 $$
    0=\left(\frac{w_x(t,X_{t-})}{2}+c(t,X_{t-})\right)w(t,X_{t-})
 $$
 yielding the stated form of $c$. To obtain $j$, observe that we should have
 $$
   w(t,X_{t-})\Delta Y_t=K_w^{-1}(t,j(t,X_{t-},\Delta Y_t)+K_w(t,X_{t-})+\Delta Y_t)-X_{t-}.
 $$
 
\end{example}
	
	As those examples illustrate, the process $C$ in Definition \ref{mm:d:price}  prices the additional martingale part of the total demand incurred by the insider's strategy whereas $J$ penalizes the jumps in $Y$. Moreover, the Doss-Lamperti transformation $K_w$ ensures that the price of the jumps is given by the units of the total demand process and is independent of the scaling factor $w$. That is, the form of the price postulated in Definition \ref{mm:d:price} allows a general penalization of martingale and jump components of insider strategy and therefore covers, as particular cases, vast majority of pricing rules considered in the literature.  This becomes even more apparent if we rewrite (\ref{mm:eq:signal_mm1})  as (note that this  is equivalent representation  due to the conditions on $w$\footnote{ Indeed, choosing $c=\frac{\tilde{c}}{w}-\frac{w_x}{2}$ and 
	\[
	j(t,x,\Delta)=K_w(t,\tilde{j}(t,x,\Delta)+x)-K_w(t,x)-\Delta,
	\]
	gives us the representation in the form of (\ref{mm:eq:signal_mm1})-(\ref{mm:e:Jdef}). Moreover, the relevant conditions in Definition \ref{mm:d:prule} are satisfied in view of the mean value theorem and  the fact that that $w>0$. 
}):
	\[
	dX_t = w (t,X_{t-}) dY^c_t+\tilde{c} (t,X_{t-})(d[Y,Y]^c_t-dt)+ \tilde{j}(t,X_{t-},\Delta Y_t) ,\quad X_0 =0,
	\]
	with  $\tilde{j}$ satisfying (\ref{mm:e:jgrowth}) for some $\Delta^*$ and  $\tilde{c}$ being locally Lipschitz.
	
If $\theta$ is absolutely continuous, $dX_t=w(t,X_t)dY_t$ and the price set by the market makers agrees with the one set in the standard literature. That is, if the insider's trading strategy is restricted to be absolutely continuous, the market price process is the same for all choices of $c$ and $j$. This implies that  (\ref{mm:eq:signal_mm1}) defines a set of pricing rules for general strategies of the insider that are consistent with the ones used in the literature under the assumption that the insider is only allowed to follow absolutely continuous strategies. Thus, if one can identify the functions $c$ and $j$ for which the optimal strategy of the insider is absolutely continuous, one recovers the equilibria obtained in the previous studies with the modification of the pricing rule given by those $c$ and $j$.  The reason for the chosen parametrization of $C$ and $J$ is that we are going to establish that these functions are identically $0$ under our parametrization. More precisely, the insider gets infinite profit  unless $c=j=0$ as shown in Theorem \ref{t:jmprice}. \footnote{This, in particular, implies that in the Example \ref{ex:BB} and Example \ref{ex:CCA} an insider can obtain infinite profit by employing trading strategies with martingale part and/or jumps.}

  In the next few remarks we discuss well-posedness  of our definitions of admissible trading strategies and the price set by the market makers as well as their direct implications.   
}       
    \begin{remark} \label{mm:r:insfilt} Note that Condition \ref{mm:d:noempty} in Definition \ref{mm:d:prule} is equivalent to no trading being admissible for the insider. Thus, $\mathcal{A}(H,w,c,j)\neq \emptyset$ for any admissible pricing rule. 
    	
    	Moreover, the strict monotonicity of $H$ in the space variable implies $H$ is invertible prior to time $1$,
    	thus, the filtration of the insider is generated by $X$ and $Z$. Note that  jumps of $Y$ can be inferred from the jumps of $X$ via (\ref{mm:eq:signal_mm1}) and the form of $J$.  Moreover, since $K_w \in C^{1,2}$ under the hypothesis on $w$, an application of Ito's formula yields
    	\[
    	dK_w(t,X_t)=dY^c_t -\half w_x(t,X_{t-})dt +K_w(t,X_{t})-K_w(t,X_{t-})+\frac{\partial}{\partial t}K_w(t,X_{t-})dt.
    	\] 
    	Thus, one can also obtain the dynamics of $Y^c$ by observing $X$. Hence,  the natural filtrations of $X$ and $Y$ coincide. This in turn implies that $(\cF^{S,Z}_t)=(\cF^{B,Z}_t)$, i.e. the insider has full information about the market. This justifies our choice of the filtration to which $\theta$ is adapted in Definition \ref{mm:d:iadm}.
    \end{remark}    
        
\begin{remark} \label{r:welldefined}
	Note that since $\theta$ is assumed to be admissible in Definition \ref{mm:d:price}, it is a semimartingale with summable jumps. This in turn implies $Y$ is a semimartingale with summable jumps as well. Therefore, the equations (\ref{mm:eq:signal_mm1})-(\ref{mm:e:Jdef}) are well-posed. In particular, the processes $C$ and $(\sum_{s\leq t}J_s)_{t\in [0,1]}$ are adapted to the  market makers' filtration $\cF^M$ and are of finite variation. Indeed, since $X$ is a strong solution and $Y$ is a semimartingale, they are both right continuous with left limits. Thus, their paths are bounded on compacts a.s.. This readily implies that $C$ is of finite variation as the functions $c, w$ and $w_x$ are continuous. 
	
	Moreover, due to the mean value theorem and the fact that $K_w^{-1}$ has a continuous derivative, we have
	\[
	|J_t|\leq \gamma(\omega)|j(t,X_{t-},\Delta Y_t)+\Delta Y_t|,
	\]
	 where $\gamma$ is bounded. Since there are only finitely many jumps of $Y$ that exceed $\Delta^*$, the above bound implies that $\sum_{t\leq 1}|J_t|<\infty$ in view of (\ref{mm:e:jgrowth}). 
\end{remark}
\begin{remark} In the standard models $Z$ is assumed to be a solution of an SDE driven by a Brownian motion. This entails that $\cF^{B,Z}$ is contained in a Brownian filtration. Therefore,  the jumps of $\theta$ are summable as soon as $\theta$ is assumed to be a semimartingale. Hence,  the first condition of the definition above reduces to the requirement that $\theta$ is a semimartingale. 
\end{remark}

Now we can formally define the market equilibrium as follows.
\begin{definition} \label{eqd} A couple $((H^{\ast}, w^*, c^*,j^*), \theta^{\ast})$
	is said to form an equilibrium if $(H^{\ast},w^*, c^*,j^*)$ is an admissible pricing rule,
	$\theta^{\ast} \in \cA(H^{\ast},w^*, c^*,j^*))$, and the following conditions are
	satisfied:
	\begin{enumerate}
		\item {\em Market efficiency condition:} given $\theta^{\ast}$,
		$(H^{\ast},w^*, c^*,j^*)$ is a rational pricing rule, i.e. it satisfies (\ref{mm:d:mm_obj}).
		\item {\em Insider
			optimality condition:} given $(H^{\ast},w^*, c^*,j^*)$, $\theta^{\ast}$ solves
		the insider optimization problem for all $z$:
		\[
		E^{0,z}[W^{\theta^{\ast}}_1] = \sup_{\theta \in \cA(H^{\ast},w^*, c^*,j^*)} E^{0,z} [W^{\theta}_1]<\infty.
		\]
	\end{enumerate}
\end{definition}

 We finish this section by summarizing our standing assumptions for the convenience of the reader.
 \begin{assumption} \label{mm:a:prule} \begin{enumerate}
 		\item  $Z$ is a continuous and $\cG$-adapted process\footnote{Recall that this implies that the regular conditional distribution of  $(X_s, Z_s; s\leq 1)$ given $X_0=0$ and $Z_0=z$ exists  in view of Theorem 44.3 in \cite{Bauer}.}.
 		\item  $f:\bbR\to\bbR$ is a  measurable and increasing function.
 		\item The pricing rule $(H,w,c,j)$ satisfies       
 		\be \label{mm:e:pdeh}
 		H_t(t,x)+\half w^2(t,x)H_{xx}(t,x)=0.
 		\ee
 	\end{enumerate} 
 \end{assumption}

In what follows we will use function     $g$ defined as 
\be \label{mm:eq:pdewg} 
g(t,x):= \frac{w_t(t,x)+\frac{w^2(t,x)}{2}w_{xx}(t,x)}{w^2(t,x)}
\ee  
for a given admissible pricing rule  $(H,w,c,j)$. Note that it is continuous due to the conditions on $w$.
      
\begin{remark}
A formal derivation of HJB equations associated with the insider optimisation problem in the case the insider's signal is Markovian as in \cite{CCD} and \cite{CDBook} will lead to (\ref{mm:e:pdeh}) and $g\equiv 0$.  Moreover, we will demonstrate that  $g\equiv 0$ is necessary for the equilibrium to exist\footnote{If the equilibrium is inconspicuous as in most of the literature, the stated PDE for $H$ will follow from the standard filtering theory.}. 
\end{remark}
The above setup uses the standard definition of equilibrium as in Back \cite{Back}. The difference lies in the generalisation of the set of admissible pricing rules that in particular includes the ones used in the current literature (see, e.g., \cite{Back-Baruch}, \cite{Cho}, \cite{CDF}, \cite{CCD}, \cite{CRH}, \cite{CFNO}, \cite{CNF}, and \cite{MSZ}). Moreover, the signal of the insider is not assumed to be Markovian{, i.e. our setting is less stringent than the current literature where $Z$ is assumed to be Markov.}
\section{Main results}

{ There are three main results in this paper:
\begin{itemize}
 \item Our first main result, Theorem \ref{t:jmprice}, shows that the insider achieves infinite profits unless the pricing rule penalises the jumps and martingale parts correctly, i.e. $c=j=0$. \footnote{Observe that this condition is not satisfied in the Examples \ref{ex:BB} and \ref{ex:CCA}. }
 \item Our second main result, Theorem \ref{t:gzero}, shows that in a wide range of models for a pricing rule to be consistent with an equilibrium it's weighting function must satisfy  (\ref{mm:eq:pdewg}) with $g=0$.
 \item Our third main result, Theorem \ref{mm:t:AC}, shows that, if the market maker chooses a pricing rule with correct penalization of martingale part/jumps of insider's strategy that is also consistent with an equilibrium in a sense of the Theorem \ref{t:gzero} (i.e.  $g=0$ in (\ref{mm:eq:pdewg})), then there exists a sequence of absolutely continuous insider trading strategies that maximizes insider's objective. This implies that one can restrict admissible trading strategies to absolutely continuous one without loss of generality as soon as pricing rule is chosen correctly.  
\end{itemize}

For the reader's convenience of a reader the main results, as well as remarks discussing their assumptions, are stated below.

\begin{theorem} \label{t:jmprice} Let $(H,w,c,j)$ be an admissible pricing rule such that $H$ satisfies Assumption \ref{mm:a:prule}. Consider  $g$ defined by (\ref{mm:eq:pdewg}) and assume that for every $z \in \bbR$ there exists an $x \in \bbR$ such that
	\be \label{mm:a:g}
	E^{0,z}\left[\int_0^{1-}\int_{\xi(t,f(Z_1))} ^{x(z)} (H(t,u)-f(Z_1))|g(t,u)|dudt\right]<\infty.
	\ee
	Assume further that the random variable $f(Z_1)$ is such that :
	\begin{itemize}
		\item 
		\be\label{mm:eq:sqint}
		E^{0,z}\left[f^2(Z_1)\right]+E^{0,z}\left[K_w^2(1, H^{-1}(1,f(Z_1)))\right]<\infty, \qquad \forall z \in \bbR,
		\ee
		\item $\lim_{z \rar -\infty}E^{0,z}[-f(Z_1)]=\lim_{z \rar \infty}E^{0,z}[f(Z_1)]=\infty$,
		\item and $\limsup_{z \rar -\infty}E^{0,z}[f(Z_1)\chf_{[f(Z_1)>k]}]<\infty$, $\liminf_{z \rar \infty}E^{0,z}[f(Z_1)\chf_{[f(Z_1)<k]}]>-\infty$. 
	\end{itemize}
	Then there exists a set $E$ such that $\bbQ(Z_0\in E)>0$ and for any $z\in E$ we have 
	\[
	\sup_{\theta \in \cA(H,w,c,j)}E^{0,z}[W_1^{\theta}]=\infty
	\]
	unless $c$ and $j$ are identically $0$.
\end{theorem}

\begin{remark}
	The condition (\ref{mm:a:g}) ensures that the value function of the insider is  bounded when she is restricted to use absolutely continuous strategies. This condition is always satisfied if insider has static information, i.e. if $Z_0=Z_1$. The assumptions on the random variable $f(Z_1)$ are quite general and are satisfied in the available literature. In particular they are satisfied in a large class of diffusion models. \end{remark}
\begin{theorem} \label{t:gzero}
	Suppose that there exists an equilibrium $((H^*,w^*),\theta^*)$, where  $H^*$  satisfies Assumption \ref{mm:a:prule}. Consider $g$ defined by (\ref{mm:eq:pdewg}) and assume $Z_0=Z_1$, $\lim_{z \rar \infty}f(z)=-\lim_{z \rar -\infty}f(z) =\infty$, and $\frac{g}{H_y}$ as a function on $[0,1]\times \bar{\bbR}$ with values in $\bar{\bbR}$ is continuous. 
	
	Assume further that there exists a set $E$ such that $\bbQ(Z_0\in E)=1$ and for all $z\in E$ there exists a continuous function $s^{f(z)}$ of finite variation such that 
	\be\label{mm:eq:max_g}
	s^{f(z)}(t):=\arg \max_x -\int_{\xi(t,f(z))}^x (H^*(t,y)-f(z)) g(t,y)dy=\arg \min_x \int_{\xi(t,f(z))}^x (H^*(t,y)-f(z)) g(t,y)dy
	\ee
	for all $t\in[0,\nu]$. Then $g\equiv 0$ for all $t\in[0,\nu)$.
\end{theorem}
\begin{theorem} \label{mm:t:AC} Suppose that 
	\be \label{mm:e:fsqint}
	E^{0,z}\left(f^2(Z_1)\right)<\infty, \, \forall z \in \bbR,
	\ee
	i.e. $f(Z_1)$ is square integrable for any initial condition of $Z$. Let  $(H,w)$ be an admissible pricing rule satisfying (\ref{mm:eq:pdewg}) and (\ref{mm:e:pdeh})  with $g=0$.
	Then  $\theta \in \cA(H,w)$ is an optimal strategy if 
	\begin{enumerate}
		\item[i)] $\theta$ is continuous and of finite variation,
		\item[ii)]   and $H(1-,X_{1-})=f(Z_1), \, P^{0,z}$-a.s.,
	\end{enumerate} 
	where 
	\[
	X_t= \int_0^t w(s,X_s)\{dB_s +d\theta_s\}.
	\]
	
	Moreover, if we further assume that 
	\be \label{mm:e:Hextra}
	E^{0,z}\left[K_w^2(1, H^{-1}(1,f(Z_1)))\right]<\infty, \qquad \forall z \in \bbR, 
	\ee
	and $M$ defined by 
	\be \label{e:defM}
	M_t:=E^{0,z}\left[K_w(1,H^{-1}(1,f(Z_1)))|\cF^Z_t\right]
	\ee
	satisfies 
	\be \label{e:MQV}
	d[M,M]_t=\tilde{\sigma}^2_t dt
	\ee 
	for some measurable process $\tilde{\sigma}$ such that 
	\be \label{e:MQVbd}
	\limsup_{t \rar 1} \tilde{\sigma}^2_t (1-t)^{\alpha-1}=0
	\ee
	for some $\alpha \in (1,2)$, 
	then for any  $\theta \in \cA(H,w)$, there exists a sequence of admissible {\em absolutely continuous} strategies, $(\theta^n)_{n \geq 1}$, such that
	\[
	E^{0,z}\left[W_1^{\theta}\right]\leq \lim_{ n \rar \infty} E^{0,z}\left[W_1^{\theta^n}\right].
	\]
\end{theorem}
\begin{remark}
	To understand the conditions (\ref{e:defM})-(\ref{e:MQVbd}) note that  $K_w(1,H^{-1}(1,f(Z_1)))$ is $P^{0,z}$-integrable, therefore $M$ is well-defined and is independent of $B$. 
	
	This process $M$ will be used by the insider to drive the market price to its fundamental value. Under the optimality conditions of the theorem above $K_w(1,H^{-1}(1,f(Z_1)))=Y_1$. Thus, $M$ corresponds to the insider's expectation of the final total demand using her own private information only. Not using public information ensures $M$ is independent of $B$. 
	
	The condition (\ref{e:MQVbd}) is in fact  an assumption on the quadratic variation of the signal and is satisfied in the Markovian framework employed in the earlier Kyle-Back models (see, among others, \cite{Danilova}, \cite{BP}, \cite{CCD},\cite{Wu}, and \cite{CCD2}).
	
\end{remark}
}
\section{Penalization in action}
{ In this section we illustrate our results with the model considered in Back and Baruch \cite{Back-Baruch} and  discussed in the Example \ref{ex:BB}. We will show that the pricing rule employed therein allows infinite profit for the insider  that uses trading strategies with martingale component and describe the correct penalization  of the martingale component that prevents such opportunities and ensures that the absolutely continuous strategies are optimal.\footnote{ One can also show that insider that uses jumps achieves infinite profit in straightforward generalization of the pricing rule of \cite{Back-Baruch} where $dX_t=\lambda(X_{t-})dY_t$. It is also possible to determine a correct penalization that makes jumps suboptimal for the insider.  To do so one needs to employ tedious calculations as in the proof of Theorem \ref{t:thetacts}. Thus, to simplify the exposition we will only demonstrate the results for martingale parts assuming continuous strategies.}

 Although this model does not directly fit to the above framework at a first sight, the exponential character of $\tau$ allows us to view the setting of Back and Baruch (and a more general version studied in \c{C}etin \cite{CRH}) as an infinite horizon version of the Kyle model. 

More precisely, given a continuous semimartingale trading strategy $\theta$ of the insider, the associated final wealth at time $\tau$ is given by
\[
W_{\tau}=Z_1\theta_{\tau} -\int_0^{\tau} S_td\theta_t -[\theta,S]_t.
\]
Back and Baruch assume that $S_t=X_t$, where
\[
dX_t=\lambda(X_s)dY_s,
\]
for some function $\lambda$ that is implicitly defined by an equation and  with the boundary conditions $\lambda(0)=\lambda(1)=0$. \c{C}etin \cite{CRH} later shows that 
\[
\lambda(x)=s'(s^{-1}(x)),
\]
where 
\[
s(x)=\int_{-\infty}^x \sqrt{\frac{r}{\pi}}\exp(-ry^2)dy.
\]
Now suppose that $\theta$ has a local martingale component driven only by $B$ and let
\[
\gamma_t:=\frac{d}{dt}[\theta,B]_t.
\]
Thus, using the independence of $\tau$, we can find the expected profit of the insider given $Z_1=z$ as 
\be \label{e:bbwealth}
E^{0,z}\left[W_{\tau}\right]=E^{0,z}\left[\int_0^{\infty}e^{-rt} (z-X_t)d\theta_t -\int_0^{\infty} e^{-rt} \lambda(X_t) (1+\gamma_t)\gamma_tdt\right].
\ee
Following \cite{CRH} define
\[
\Psi(x)=\int_z^x\frac{y-z}{\lambda(y)}dy
\]
and observe that
\be\label{e:ex_psi}
\frac{\lambda^2}{2}\Psi''-r\Psi=0,
\ee
and that 
\be\label{e:ex_psi_second}
\Psi''(x)=\frac{1}{\lambda(x)}+2r \frac{x-z}{\lambda^2(x)}s^{-1}(x).
\ee
Therefore,
\[
e^{-rt}\Psi(X_t)=\Psi(x)+\int_0^t e^{-rs}\Psi'(X_s)dX_s+\half\int_0^t e^{-rs} \lambda^2(X_s)\gamma_s(2+\gamma_s)\Psi''(X_s)ds,
\]
which in turn yields
\[
\begin{split}
E^{0,z}\left[W_{\tau}\right]=&\Psi(x)-\lim_{t\rar \infty}E^{0,z}\left[e^{-rt}\Psi(X_t)\right] \\
&+E^{0,z}\left[\int_0^{\infty} e^{-rs} \left(rs^{-1}(X_s)(X_s-z)\gamma_s(2+\gamma_s)-\lambda(X_s)\frac{\gamma_s^2}{2}\right)ds\right].
\end{split}
\]
We shall now show that the insider will obtain infinite profits given this pricing rule by appropriately choosing $\gamma$. To this end it suffices to consider the case $z=0$.

Note that
\[
\gamma \mapsto rs^{-1}(X_s)X_s\gamma(2+\gamma)-\lambda(X_s)\frac{\gamma^2}{2}
\]
is strictly convex if $2r s^{-1}(X_s)X_s-\lambda(X_s)>0$. Since $s^{-1}(1)=\infty$ and $\lambda(x)=s'(s^{-1}(x)$, it is easy to show the existence of some $p^*$ such that $2r s^{-1}(x)x-\lambda(x)>1$ for $x>p^*$. Also observe that $\Psi(1-\eps) <\infty$ for $\eps>0$. Thus, the insider's optimal strategy is to keep the price process between $p^*$ and $1-\eps$  by using a constant $\gamma$ process provided $rs^{-1}(X_s)X_s\gamma(2+\gamma)-\lambda(X_s)\frac{\gamma^2}{2}>1$ all the time. This can be achieved by using a recurrent transformation with $h(x)= (s(x)-s(p^*))(s(1-\eps)-s(x))$ (cf. Theorem 3.1 in \cite{rtr}) or an infinite horizon version of the method used in Stage 3 of the proof of Theorem \ref{t:jmprice} in this paper. Scaling this $\gamma$ will lead to infinite profits in the limit. 
To penalize correctly the martingale components of a given continuous strategy the market makers should choose $X$ such that
\be\label{e:ex_X}
dX_t= \lambda(X_{t})dY_t+\frac{\lambda'(X_{t})\lambda(X_{t})}{2}(d[Y,Y]_t-dt).
\ee
This choice above will lead the insider to use continuous strategies with no martingale component. Indeed, in this case
\[
\begin{split}
   e^{-rt}\Psi(X_t)=&\Psi(x)-\int_0^t re^{-rs}\Psi(X_s)ds+\int_0^t e^{-rs}\Psi'(X_s)dX_s+\half\int_0^t e^{-rs} \lambda^2(X_s)\Psi''(X_s)d[Y,Y]_s\\
   \stackrel{(\ref{e:ex_psi})}{=}&\Psi(x)+\int_0^t e^{-rs}\Psi'(X_s)dX_s+\half\int_0^t e^{-rs} \lambda^2(X_s)\Psi''(X_s)(d[Y,Y]_s-ds)\\
    \stackrel{(\ref{e:ex_X})}{=}&\Psi(x)+\int_0^t e^{-rs}\Psi'(X_s)\lambda(X_{s})dY_t\\
    &+\half\int_0^t e^{-rs}\left(\Psi'(X_s)\lambda'(X_{s})\lambda(X_{s})+ \lambda^2(X_s)\Psi''(X_s)\right)(d[Y,Y]_s-ds)\\
    \stackrel{(\ref{e:ex_psi_second})}{=}&\Psi(x)+\int_0^t e^{-rs}\Psi'(X_s)\lambda(X_{s})dY_t\\
    &+\half\int_0^t e^{-rs}\left((X_s-z)\left(\lambda'(X_{s})+2r s^{-1}(X_s)\right)+ \lambda(X_s)\right)(d[Y,Y]_s-ds)\\
   =&\Psi(x)+\int_0^t e^{-rs}(X_s-z)dY_s+\half\int_0^t e^{-rs} \lambda(X_s)(d[Y,Y]_s-ds).
\end{split}
\]
Observe also that calculations similar to the ones that lead to (\ref{e:bbwealth}) will yield
$$
E^{0,z}\left[W_{\tau}\right]=E^{0,z}\left[\int_0^{\infty}e^{-rt} (z-X_t)d\theta_t -\int_0^{\infty} e^{-rt} \lambda(X_t) \left(d[Y,Y]_t-d[B,Y]_t\right)\right].
$$
and therefore
\[
\begin{split}
E^{0,z}\left[W_{\tau}\right]=&\Psi(x)-\lim_{t\rar \infty}E^{0,z}\left[e^{-rt}\Psi(X_t)\right] \\
&+E^{0,z}\left[\half\int_0^\infty e^{-rt} \lambda(X_t)(d[Y,Y]_t-dt)-\int_0^{\infty} e^{-rt} \lambda(X_t) \left(d[Y,Y]_t-d[B,Y]_t\right)\right]\\
=&\Psi(x)-\lim_{t\rar \infty}E^{0,z}\left[e^{-rt}\Psi(X_t)\right] -\half E^{0,z}\left[\int_0^\infty e^{-rt} \lambda(X_t)d[Y-B,Y-B]_t\right].
\end{split}
\]
This shows that additional martingale component is strictly suboptimal.}

\section{On equilibrium pricing rules and optimal strategies}
 In this section we will { prove Theorems \ref{t:jmprice} -- \ref{mm:t:AC}. In particular, we will }  show that  in the representation (\ref{mm:eq:signal_mm1}) $c=j=0$ and $w$ satisfying (\ref{mm:eq:pdewg}) { with $g=0$} are necessary conditions for { a} pricing rule to be compatible with { an} equilibrium, since any other choice of market maker's weighting of the signal will result in the infinite profit for the insider  and/or make equilibrium impossible.{ We will also show that when a pricing rule satisfies those necessary conditions trading strategies of the insider can be restricted to absolutely continuous ones.}
      
Our first theorem computes the expected final wealth of the insider in our general setup. We will use this representation to solve the optimisation problem for the insider. In particular this representation will provide an upper bound on the value function when $g$ vanishes and the trading strategies are continuous. 
\begin{theorem} \label{mm:t:iWealth} Let  $(H,w,c,j)$ be an admissible pricing rule  such that $H$  satisfies Assumption \ref{mm:a:prule}.  Assume $\theta\in \mathcal{A}(H,w,c,j)$. Then
	\bean
	E^{0,z}\left[W_1^{\theta}\right]&=&E^{0,z}\left[\Psi^{f(Z_1)}(0,0)-\Psi^{f(Z_1)}(1-,X_{1-})-\half \int_0^{1-}w(t,X_{t-})H_x(t,X_{t-})d[\theta,\theta]^c_t\right.\nn \\
	&&+\int_0^{1-}\left(H(t,X_{t-})-a\right) c(t,X_{t-})(d[Y,Y]^c_t-dt) -\int_0^{1-}\int_{\xi(t,a)} ^{X_{t-}} (H(t,u)-a)g(t,u)dudt\\
	&&+\left. \sum_{0<t<1}\left\{\Psi^{f(Z_1)}(t,X_t)-\Psi^{f(Z_1)}(t,X_{t-})-(H(t,X_{t})-f(Z_1))\Delta\theta_t\right\} \right],
	\eean
	where
		\begin{equation}\label{mm:e:generalG_a}
	\Psi^a(t,x):=\int_{\xi(t,a)} ^x \frac{H(t,u)-a}{w(t,u)}du+\frac{1}{2}\int_t^1H_x(s,\xi(s,a))w(s,\xi(s,a))ds,
	\end{equation}
	the function $g$ is given by (\ref{mm:eq:pdewg}), and  $\xi(t,a)$ is the unique solution of $H(t,\xi(t,a))=a$.
	
	Moreover,
	\be \label{mm:eq:uljumps}
	\begin{split}
		\Delta \Psi^{a}(t,X_t)-(H(t,X_t)-a)\Delta\theta_t\leq (H(t,X_{t})-a)j(t,X_{t-},\Delta Y_t)\\
		\Delta \Psi^{a}(t,X_t)-(H(t,X_t)-a)\Delta\theta_t\geq (H(t,X_{t-})-a)j(t,X_{t-},\Delta Y_t)-\Delta H(t,X_t)\Delta \theta_t.
	\end{split}
	\ee
\end{theorem}
\begin{proof}
	Using Ito's formula for general semimartingales (see, e.g. Theorem II.32 in \cite{Pro}) we obtain
	\[
	dH(t, X_t)=H_x(t,X_{t-})w(t,X_{t-})dY^c_t +dFV_t,
\]
	where $FV$ is of finite variation, in view of Remark \ref{r:welldefined}. Therefore,
	\be
	[\theta, S]^c_t =\int_0^t  H_x(s,X_{s-})w(s,X_{s-})\left\{d[B,\theta]_s+d[\theta,\theta]^c_s\right\}. \label{mm:eq:THQV}
	\ee
	Moreover, integrating (\ref{mm:eq:insW}) by parts (see Corollary 2 of Theorem  II.22 in \cite{Pro}) we get
	\be \label{mm:eq:Wibp}
	W^{\theta}_1=f(Z_1)\theta_{1-}-\int_0^{1-}H(t,X_{t-}))d\theta_t -[\theta,H(\cdot, X)]_{1-}
	\ee
	since the jumps of $\theta$ are summable. Moreover, direct calculations lead to 
	\begin{equation} \label{mm:eq:pdepsi}
	\Psi^a_t+\frac{1}{2}w(t,x)^2\Psi^a_{xx}=-\int_{\xi(t,a)} ^x (H(t,u)-a)g(t,u)du.
	\end{equation}
 Ito's formula in conjunction with above  yields
	\bean
	\Psi^a(1-,X_{1-})&=&\Psi^a(0,0)+ \int_0^{1-}H(t,X_{t-})(dB_t +d\theta_t) -a (B_1 +\theta_{1-})\\
	&&+ \half \int_0^{1-}w(t,X_{t-})H_x(t,X_{t-})(d[Y,Y]^c_t-dt) \\
	&&+ \sum_{0<t<1}\left\{\Psi^{a}(t,X_t)-\Psi^a(t,X_{t-})-(H(t,X_{t-})-a)\Delta\theta_t\right\}\\
	&&+\int_0^{1-}\left(H(t,X_{t-})-a\right) c(t,X_{t-})(d[Y,Y]^c_t-dt) -\int_0^{1-}\int_{\xi(t,a)} ^{X_{t-}} (H(t,u)-a)g(t,u)dudt
	\eean
	Combining the above and (\ref{mm:eq:Wibp}) and noting that the stochastic integral with respect to $B$ is a true martingale we deduce
	\bean
	E^{0,z}\left[W_1^{\theta}\right]&=&E^{0,z}\left[\Psi^{f(Z_1)}(0,0)-\Psi^{f(Z_1)}(1-,X_{1-})-\half \int_0^{1-}w(t,X_{t-})H_x(t,X_{t-})d[\theta,\theta]^c_t\right.\nn \\
	&&+\int_0^{1-}\left(H(t,X_{t-})-a\right) c(t,X_{t-})(d[Y,Y]^c_t-dt) -\int_0^{1-}\int_{\xi(t,a)} ^{X_{t-}} (H(t,u)-a)g(t,u)dudt\\
	&&+\left. \sum_{0<t<1}\left\{\Psi^{f(Z_1)}(t,X_t)-\Psi^{f(Z_1)}(t,X_{t-})-(H(t,X_{t})-f(Z_1))\Delta\theta_t\right\} \right].
	\eean
	Note that since $w$ is positive  and $H$ is increasing, we have
	\bean
	\Psi^{a}(t,X_t)-\Psi^{a}(t,X_{t-})-(H(t,X_t)-a)\Delta\theta_t&=&\int_{X_{t-}}^{X_t}\frac{H(t,u)-a}{w(t,u)}du -(H(t,X_t)-a)\Delta\theta_t \\
	&\leq&(H(t,X_t)-a)\int_{X_{t-}}^{X_t}\frac{1}{w(t,u)}du -(H(t,X_t)-a)\Delta\theta_t \\
	&=&(H(t,X_t)-a)j(t,X_{t-},\Delta Y_t).
	\eean
	Similarly,
	\[
	\Psi^{a}(t,X_t)-\Psi^{a}(t,X_{t-})-(H(t,X_t)-a)\Delta\theta_t\geq (H(t,X_{t-})-a)j(t,X_{t-},\Delta Y_t)-\Delta H(t,X_t)\Delta \theta_t.
	\]
\end{proof}
\begin{remark}
	\label{r:absfinval} The representation of the expected profit given by the above theorem shows that the absolutely continuous strategies deliver expected wealth bounded by $E^{0,z}[\Psi^{f(Z_1)}(0,0)]$ when $g\equiv 0$. Similarly, if the  optimisation problem 
	\[
	\inf_X E^{0,z}\left[\int_0^{1-}\int_{\xi(t,f(Z_1))} ^{X_t} (H(t,u)-f(Z_1))|g(t,u)|dudt\right]
	\]
	has a finite value, the value function of the insider is also bounded when she is restricted to use absolutely continuous strategies. Theorem \ref{t:jmprice} that we will prove next imposes conditions  sufficient for this to hold.
\end{remark}

\begin{proof}[Proof of Theorem \ref{t:jmprice}]
	Suppose $c(t,x)\neq 0$ for some $t<1$ and $x$. Since $c$ is continuous, there exist $\nu_1<\nu_2<1$ and $x_1<x_2$ such that $|c(t,x)|>\eps$ for some $\eps>0$ on $[\nu_1,\nu_2]\times[x_1,x_2]$. Moreover, by the continuity of $K_w^{-1}$ there exists $t_1<t_2<1$ and $y_1<y_2$ such that $K_w^{-1}(t,y) \in [x_1,x_2]$ for  all $(t,y)\in [t_1,t_2]\times [y_1,y_2]$. 
	
	We shall construct a continuous trading strategy to achieve arbitrarily large profits for some realisation of $Z_1$. This construction will be done in three stages.  The first stage will utilise Lemma \ref{mm:l:appr} to bring $X$ inside $[K_w^{-1}(t_1,y_1),[K_w^{-1}(t_1,y_2)]$ at time $t_1$.  The second stage will keep $K_w(t,X_t)$ inside the interval $[y_1,y_2]$ with arbitrarily large quadratic variation. The final stage will keep $X$ bounded  up to time $1$. 
	
	Observe that for any continuous semimartingale $\theta$ and $G(t,x):=\int_0^{x}g(t,y)dy$
	\[
	dK_w(t,X_t)=dY_t+ c(t,X_t)(d[Y,Y]_t-dt)-G(t,X_t)dt.
	\] 

\noindent
{\bf Stage 1:} To obtain a bounded $X^{\eps}$ satisfying $K_w(t_1,X^{\eps}_{t_1})\in (y_1,y_2)$ apply Lemma \ref{mm:l:appr} to 
\[
x(t)=\frac{(K_w^{-1}(t_1,y_1)+K_w^{-1}(t_1,y_2))t}{2 t_1} \mbox{ and } \eps=\frac{(K_w^{-1}(t_1,y_2)-K_w^{-1}(t_1,y_1))}{4}.
\]
Set $X=X^{\eps}$ on $[0,t_1]$.

\noindent
{\bf Stage 2:}  Fix a $y \in (y_1,y_2)$.  Consider interval $[t_1, t_2]$ and  the solutions of
\[
dR_t=(b+1)dB_t +(b+1)^2 \left(\frac{1}{R_t-y_1}\chf_{[R_t \leq y]}-\frac{1}{y_2-R_t}\chf_{[R_t>y]}\right)dt+ (b^2+2b)c(t, K_w^{-1}(t,R_t))dt.
\]
and observe that pathwise uniqueness holds  until the exit time from $(y_1,y_2)$ since $c$ is locally Lipschitz and $K_w^{-1}$ is continuously differentiable.  Thus, if we can show the existence of a weak solution that never exits $(y_1,y_2)$, we will arrive at a strong solution that stays in $(y_1,y_2)$. Indeed,  since $c(t,K_w^{-1}(t,x))$ is bounded for all $(t,x) \in (t_1,t_2)\times (y_1,y_2)$, by means of Girsanov transformation, weak solutions of above are the same as those of
\be \label{mm:eq:neverexit}
dU_t=qd\beta_t + q^2\left(\frac{1}{U_t-y_1}\chf_{[U_t \leq y]}-\frac{1}{y_2-U_t}\chf_{[U_t>y]}\right)dt,
\ee
which are unique in law and never exit $(y_1,y_2)$ by Proposition 3.1 in \cite{rtr}. Define $X_t:=K_w^{-1}(t,R_t)$ and observe that 
\[
d\theta_t=b dB_t +  (b+1)^2 \left(\frac{1}{R_t-y_1}\chf_{[R_t \leq y]}-\frac{1}{y_2-R_t}\chf_{[R_t>y]}\right)dt + G(t,K_w^{-1}(t,R_t))dt.
\]

\noindent
{\bf Stage 3:} Finally consider the interval $(t_2,1]$. Apply Lemma \ref{mm:l:appr} to $x(t)=X_{t_2}$ and $\eps$ as before to get $|X^{\eps}|$  and set $X=X^{\eps}$. 

Observe that $X$ constructed above  is bounded by a determinstic constant, which in turn implies the boundedness of $H(t,X_t)$. Thus, $\theta \in \cA(H,w,c,j)$. Recall from Theorem \ref{mm:t:iWealth} that

\bean
E^{0,z}\left[W_1^{\theta}\right]&\geq&E^{0,z}\left[-\Psi^{f(Z_1)}(1-,X_{1-})-\frac{b^2}{2}\int_{t_1}^{t_2}w(t,X_{t-})H_x(t,X_{t-})dt\right.\nn \\
&&+\left.(b^2+2b)\int_{t_1}^{t_2}\left(H(t,X_{t-})-f(Z_1)\right) c(t,X_{t-})dt -\int_0^{1-}\int_{\xi(t,f(Z_1))} ^{X_{t-}} (H(t,u)-f(Z_1))g(t,u)dudt\right]
\eean
 since  $\Psi^{f(Z_1)}(0,0)\geq 0$ and $d[\theta,\theta]_t=0$ for $t \in [0,1]\backslash [t_1,t_2]$. Moreover, as $dK_w(1,u)=\frac{1}{w(1,u)}$,
\[
\Psi^{f(Z_1)}(1-,X_{1-})\leq (H(1,X_{1-})-f(Z_1))(K_w(1,X_{1-})-K_w(1,\xi(1,f(Z_1))))\nn
\]
since $H$ and $K_w$ are increasing functions. Since $X$ is bounded and (\ref{mm:eq:sqint}) holds, we deduce 
\be \label{mm:eq:psibound}
E^{0,z}[\Psi^{f(Z_1)}(1-,X_{1-})]<\ell_1(z)<\infty. 
\ee
Also observe that
\bean
\int_0^{1-}\int_{\xi(t,f(Z_1))} ^{X_{t-}} (H(t,u)-f(Z_1))|g(t,u)|dudt&=&\int_0^{1-}\int_{\xi(t,f(Z_1))} ^{x(z)} (H(t,u)-f(Z_1))|g(t,u)|dudt\\
&&+\int_0^{1-}\int_{x(z)} ^{X_{t-}} (H(t,u)-f(Z_1))|g(t,u)|dudt.
\eean
Thus, for some constant $\ell_2(z)$ independent of $b$ due to Assumption \ref{mm:a:g} and $X$ taking  values in a bounded interval, 
\be \label{mm:e:gintegrable}
E^{0,z}\left[\int_0^{1-}\int_{\xi(t,f(Z_1))} ^{X_{t-}} (H(t,u)-f(Z_1))|g(t,u)|dudt\right]\leq \ell_2(z) <\infty.
\ee

Therefore,
\bean
E^{0,z}\left[W_1^{\theta}\right]&\geq& \ell(z) +E^{0,z}\left[\int_{t_1}^{t_2}\left((b^2+2b)\left(H(t,X_{t-})-f(Z_1)\right) c(t,X_{t})-\frac{b^2}{2}w(t,X_{t-})H_x(t,X_{t-})\right)dt\right]\\
&\geq& \ell(z) + b^2(m_1+m E^{0,z}[f(Z_1)]+(M-m)E^{0,z}[f(Z_1)\chf_{[f(Z_1)<0]}]) \\
&&+ b(m_3+2m E^{0,z}[f(Z_1)]+2(M-m)E^{0,z}[f(Z_1)\chf_{[f(Z_1)<0]}])\\
&=&  \ell(z) + b^2(m_1+M E^{0,z}[f(Z_1)]+(m-M)E^{0,z}[f(Z_1)\chf_{[f(Z_1)\geq 0]}]) \\
&&+ b(m_3+2M E^{0,z}[f(Z_1)]+2(m-M)E^{0,z}[f(Z_1)\chf_{[f(Z_1)\geq 0]}]),
\eean
where the constants $m$ and $M$  such that $mM>0$ and $M\geq -\int_{t_1}^{t_2}c(t,X_t)dt\geq m$ exist due to the continuity of $c$, boundedness of $X$ and that $c(t,X_t)$ is bounded away from $0$ on $[t_1,t_2]$ by construction.

Observe that  the coefficient of $b^2$ in above  can be made positive for large enough $z$ (resp. small enough $z$) if $m>0$ (resp. if $M<0$)  due to our assumption on the random variable $f(Z_1)$. This implies that insider's wealth can be made arbitrarily large for such $z$ by making $b$ arbitrarily large. This yields the claim that $c$ must be $0$ for insider's profit to be finite.

Next, suppose $c\equiv0$, but  $j(t,x,\kappa)\neq 0$ for some $t<1$ and $x,\kappa$.  Without loss of generality, assume $j(t,x,\kappa)> 0$ (the proof in the case $j(t,x,\kappa)<0$ is similar). Since $j$ is continuous, there exist $t_1<t_2<1$, $x_1<x_2$, and $\kappa_1<\kappa_2$ such that $j(t,x,\kappa)>\delta$ for some $\delta>0$ on $[t_1,t_2]\times[x_1,x_2]\times[\kappa_1,\kappa_2]$. 

We will construct a strategy that achieves an arbitrarily large profit for some realisations of $Z_1$. This will be again done in three stages: first, we will bring $X$ inside the interval $[x_1,x_2]$ at time $t_1$ via Lemma \ref{mm:l:appr}. On the interval $[t_1,t_2]$ we will construct a process with an arbitrary number of jumps each of which will give positive contribution to the final utility. Finally, we will keep $X$  in the interval $[x_1,x_2]$ after time $t_2$.

Fix $0<\eps< \frac{x_2-x_1}{2}$ 
\begin{itemize}
\item[Stage 1: ] On the interval $[0,t_1)$  let $x(t)=\frac{x_1+x_2}{2t_1}t$ and apply Lemma \ref{mm:l:appr} with $\eps$ as above to obtain a bounded process $X$ such that $X_{t_1 -}\in (x_1,x_2)$. The associated $\theta^{\eps}$ will be used as insider's strategy on $[0,t_1)$. 
\item[Stage 2: ] Next, we iteratively construct the process of jumps on $[t_1,t_2]$. 

To this end, consider $s_i=t_1+i\frac{t_2-t_1}{n}$, for $i=0,n$ and set $\theta _{t_1}=\theta _{s_0}=\theta _{s_0-}+\kappa_1$. Observe that $j(s_0,X_{s_0-}, \kappa_1)>\delta$, and $X_{s_0}= K_w^{-1}(s_0,j(s_0,X_{s_0-}, \kappa_1)+K_w(s_0,X_{s_0-})+\kappa_1)$. 
      
      Suppose we already constructed the process $\theta$ (and $X$) on $[0,s_i]$ and $i<n$, then  on the interval $(s_i, s_{i+1})$ consider $x(t)=X_{s_i}+\frac{\frac{x_1+x_2}{2}-X_{s_i}}{s_{i+1}-s_i}(t-s_i)$ and apply Lemma \ref{mm:l:appr} with $\eps$ as above. Similar to Stage 1 the associated $\theta^{\eps}$ will be used as the trading strategy and $X$  satisfies $X_{s_{i+1}-}\in[x_1,x_2]$. Finally, for $i< n-1$, set $\theta _{s_{i+1}}=\theta _{s_{i+1}-}+\kappa_1$. 
     We again obtain $j(s_{i+1},X_{s_{i+1}-}, \kappa_1)>\delta$,
      and $X_{s_{i+1}}= K_w^{-1}(s_{i+1}, j(s_{i+1},X_{s_{i+1}-}, \kappa_1)+K_w(s_{i+1},X_{s_{i+1}-})+\kappa_1)$. 
      
      If $i=n-1$, define  $\theta _{s_{i+1}}=\theta _{s_{i+1}-}$ and $X_{s_{i+1}}=X _{s_{i+1}-}$.
      
      Thus, we constructed a process $\theta$ with $n$ jumps such that $[\theta,\theta]^c\equiv 0$.
\item[Stage 3: ] On the interval $[t_2,1]$ use Lemma \ref{mm:l:appr} to construct $X$ to stay in the interval $[x_1,x_2]$ by using $x(t)=\frac{x_1+x_2}{2}$ with the same $\eps$. 
\end{itemize}

Thus, we constructed a process $\theta$ and $X$ with $n$ jumps such that $[\theta,\theta]^c\equiv 0$ and $X$ is taking values in $(m, M)$ for some $m<M$.  Since $X$ is bounded, $H(t,X_t)$ is also bounded and therefore  $\theta \in \cA(H,w,c,j)$. Moreover, due to  Theorem  \ref{mm:t:iWealth}, (\ref{mm:eq:uljumps}) and the fact that $\Psi^{f(Z_1)}(0,0)\geq 0$ we have

\bean
	E^{0,z}\left[W_1^{\theta}\right]& \geq& E^{0,z}\left[-\Psi^{f(Z_1)}(1-,X_{1-}) -\int_0^{1-}\int_{\xi(t,f(Z_1))} ^{X_{t-}} (H(t,u)-f(Z_1))g(t,u)dudt\right.\nn \\
	&&+\left. \sum_{0<t<1}\left\{(H(t,X_{t-})-f(Z_1))j(t,X_{t-},\Delta Y_t)-\Delta H(t,X_t)\Delta \theta_t\right\} \right].\nn
\eean
 Using the computations that led to (\ref{mm:eq:psibound})and (\ref{mm:e:gintegrable}) we obtain
$$
   E^{0,z}\left[\Psi^{f(Z_1)}(1-,X_{1-})+\int_0^{1-}\int_{\xi(t,f(Z_1))} ^{X_{t-}} (H(t,u)-f(Z_1))|g(t,u)|dudt\right]\leq \ell_1(z)<\infty, \quad \forall z.
$$
Note that the constant $\ell_1$ is independent of $n$.

Furthermore, since the jumps of $\theta$ are of size $\kappa_1$, the jumps of $X$ are uniformly bounded and  jumps occur only at $s_i$, we have $\Delta H(t,X_t)\Delta \theta_t\leq  \ell_2<\infty.$ 

Combining the above estimates with the expression for wealth, we get
$$
E^{0,z}\left[W_1^{\theta}\right]\geq E^{0,z}\left[\sum_{i=0}^{n-1}\left\{(H(s_i,X_{s_i-})-f(Z_1))j(t,X_{s_i-},\kappa_1)-\ell_2\right\} \right]-\ell_1(z),
$$
with $l_2$ being independent of $z$.

Let $j_m=\min_{(t,x)\in[t_1,t_2]\times[x_1,x_2]}j(t,x,\kappa_1)>\delta$, $\infty>j_M=\max_{(t,x)\in[t_1,t_2]\times[x_1,x_2]}j(t,x,\kappa_1)$, $h_m=\min_{(t,x)\in[t_1,t_2]\times[x_1,x_2]}H(t,x)>-\infty$, $h_M=\max_{(t,x)\in[t_1,t_2]\times[x_1,x_2]}H(t,x)<\infty$. We have
\begin{eqnarray*}
 E^{0,z}\left[(H(s_i,X_{s_i-})-f(Z_1))j(t,X_{s_i-},\kappa_1) \right]&\geq&E^{0,z}\left[(h_m-f(Z_1))j(t,X_{s_i-},\kappa_1) \right]\\
            &=&E^{0,z}\left[(h_m-f(Z_1))j(t,X_{s_i-},\kappa_1)\chf_{[h_m<f(Z_1)]} \right]\\
           &&+E^{0,z}\left[(h_m-f(Z_1))j(t,X_{s_i-},\kappa_1)\chf_{[h_m\geq f(Z_1)]} \right]\\
           & \geq &(j_M-j_m) \left(\min\{0,h_m\}-E^{0,z}\left[f(Z_1)\chf_{[h_m<f(Z_1)]}\right]\right)\\
           && +j_m \left(h_m -E^{0,z}\left[f(Z_1)\right]\right).           
 \end{eqnarray*}
 Since $\lim_{z \rar -\infty}E^{0,z}[f(Z_1)]=-\infty$ and $\limsup_{z \rar -\infty}E^{0,z}[f(Z_1)\chf_{[f(Z_1)>k]}]<\infty$, there exists a constant $z_1$ such that for any $z<z_1$ we will have 
 $$
   E^{0,z}\left[(H(s_i,X_{s_i-})-f(Z_1))j(t,X_{s_i-},\kappa_1) \right]-\ell_2>1
 $$
 for all $i$. Thus we will have $E^{0,z}\left[W_1^{\theta}\right]\geq n-\ell_1(z),$ and letting $n\rar\infty$ will complete the proof.
\end{proof}

The following theorem allows a glimpse into the optimal strategy of the insider. It shows that if the pricing rule satisfies $c=j=0$, it is not optimal for the insider to use jumps.
\begin{theorem} \label{t:thetacts}
	Suppose that there exists an equilibrium $((H^*,w^*),\theta^*)$  such that $H$ satisfies Assumption \ref{mm:a:prule}. Consider $g$ defined by (\ref{mm:eq:pdewg}) and assume $Z_0=Z_1$, $\lim_{z \rar \infty}f(z)=-\lim_{z \rar -\infty}f(z) =\infty$, and $\frac{g}{H_y}$ as a function on $[0,1]\times \bar{\bbR}$ with values in $\bar{\bbR}$ is continuous.  Then $P^{0,z}(\om:\theta^*_t(\om) \in C([0,1) ))=1$.
\end{theorem}
The proof of this theorem is postponed to the Appendix. It relies on the following lemma that will also be useful in the proof of Theorem \ref{t:gzero}. This lemma shows that only a class of weighting functions that satisfy a further condition on $g$ can be supported in an equilibrium. Thus it allows us to restrict the set of admissible pricing rules further.  
\begin{lemma}\label{mm:l:gbbelow}
	Suppose $Z_0=Z_1$, $\lim_{z \rar \infty}f(z)=-\lim_{z \rar -\infty}f(z) =\infty$, and there exists an equilibrium $((H,w),\theta)$ such that $H$ satisfies Assumption \ref{mm:a:prule}. Consider $g$ defined by (\ref{mm:eq:pdewg}) and assume  that $\frac{g}{H_y}$ as a function on $[0,1]\times \bar{\bbR}$ with values in $\bar{\bbR}$ is continuous. Then  $\frac{g}{H_y}$ is bounded from below on $[0,1]\times \bbR$. 
\end{lemma}
\begin{proof}
	Let us first show that the existence of equilibrium implies $H(t,\infty)=-H(t,-\infty)=\infty$ for all $t \in [0,1]$. Indeed, if there exists $\hat{t}$ such that $H(\hat{t},x)\geq h$ for all $x \in \bbR$, then for all $s \leq \hat{t}$ we have $H(s,x)\geq h$ for all $x \in \bbR$. This follows from the fact that $Y$ is a Brownian motion in the market maker's filtration in equilibrium and, thus,  the random variable $X_t$ has full support due to Remark \ref{mm:r:Kwubdd}. That $Y$ is a Brownian motion is a consequence of rationality, $H_x>0$ and $H$ satisfies (\ref{mm:e:pdeh}).
	
	However, uniform boundedness of the price from below on $[0,\hat{t}]$ gives  an unbounded profit for the insider contradicting the definition of equilibrium.  Indeed, since $f$ is unbounded from below there exists a $z \in \bbR$ such that $f(z)< h$.  Consider the trading strategy
	\[
	d\theta_t=-n \chf_{[0,\hat{t}]}(t)dt
	\]
	and note that integrating by parts the associated final wealth we obtain
	\[
	W_1=-n\int_0^{\hat{t}}\left(f(z)-H(t,X_t)\right)dt\geq n (h-f(z))\hat{t}\rar\infty \mbox{ as } n \rar \infty.
		\]
		
		Therefore, we can assume that $H(t,\infty)=-H(t,-\infty)=\infty$ for all $t \in [0,1]$.
		
	Next denote $\frac{g}{H_y}(t,H^{-1}(t,x))$  by $\tilde{g}(t,x)$  and observe that 
	\[
	\int_{\xi(t,a)}^{\infty}(H(t,u)-a) g(t,u)du=\int_{a}^{\infty}(u-a) \tilde{g}(t,u)du.
	\]
	
	Suppose that for all $t \in [0,1]$ we have $\lim_{u \rar \infty} \frac{g}{H_y}(t,H^{-1}(t,u))=\lim_{u \rar \infty} \tilde{g}(t,u)\geq 0$ as well as $\lim_{u \rar -\infty} \tilde{g}(t,u)\geq 0$, and  consider  $A_n:=\{(t,u)\in [0,1]\times \bbR:\tilde{g}(t,u)\leq -n\}$. Clearly, $A_n$ is closed for each $n\geq 1$. Moreover, it is also bounded. Indeed, suppose there exists a sequence $(t_m,u_m)_{m\geq 1}\subset A_n$ such that $\lim_{m\rar \infty}t_m=t\leq 1$ and $\lim_{m \rar \infty}u_m=\infty$ or $-\infty$. Then, $\lim_{m\rar \infty}\tilde{g}(t_m,u_m)\geq 0$ due to the joint continuity of $\tilde{g}$, which is a contradiction. Thus $A_n$s are  compact and their intersection would be nonempty by the nested set property if  $\tilde{g}$ were not bounded from below.  However, if $(\hat{t},\hat{u})\in [0,1]\times \bbR$ belongs to the intersection, $\tilde{g}(\hat{t},\hat{u})=-\infty$. Thus, the intersection must be empty and, therefore, $\tilde{g}$ is bounded from below. 
	
	Next suppose that there  exists a $\hat{t}$ such that either $\lim_{u \rar \infty} \tilde{g}(\hat{t},u)< 0$ or $\lim_{u \rar -\infty} \tilde{g}(\hat{t},u)< 0$. Without loss of generality assume the former and observe that this leads to $ \tilde{g}(t,x)< -c$ for all $(t,x)$ in  $[t_1,t_2]\times [x_1,\infty)$ for some $c>0$ and $t_1, t_2$ in $[0,1]$ and $x_1 \in \bbR$ due to the joint continuity of $\tilde{g}$. Note that  $x_1$ can be assumed to satisfy $H(t,x_1)\geq f(z)$ for all $t \in [t_1,t_2]$ due to the continuity  and monotonicity of $H$.
	
	Let $x^n:[0,1] \to \bbR$ be the piecewise linear function defined by $x^n(0)=0, x^n(t_1)=x_1 +\half, x^n(\frac{3t_1+t_2}{4})=x_1 + n+\half, x^n(t_2)=x_1+ \half, x^n(1)=f(z)$ and $x^n(t)=x_1+ n+ \half$ for all $t \in [\frac{3t_1 +t_2}{4},\frac{t_1+3t_2}{4}]$. Consider $\eps =\frac{1}{4}$ and an application of Lemma \ref{mm:l:appr} yields an existence of an admissible strategy $\theta^n$ that is continuous and of finite variation and satisfies $\sup_{t\in [0,1]}|X^n_t-x^n(t)|<\frac{1}{4}$, where $X^n$ is as in the Lemma. The wealth associated with this strategy is given by
	\[
	W^n_1:=\Psi^{f(z)}(0,0)-E^{0,z}\left[\int_0^1\int_{f(z)}^{H(t,X^n_t)}(u-f(z))\tilde{g}(t,u)dudt\right]
	\]
	by Theorem \ref{mm:t:iWealth}. Since $X^n$ is uniformly bounded on $[0,t_1]$ and$[t_2,1]$, we only need to consider the above integral on $[t_1,t_2]$. Moreover, continuity of $\tilde{g}$ implies
	\bean
	W^n_1 &\geq&\ell - E^{0,z}\left[\int_{t_1}^{t_2}\int_{x_1}^{H(t,X^n_t)}(u-f(z))\tilde{g}(t,u)dudt\right]\\ &\geq &\ell - E^{0,z}\left[\int_{\frac{3t_1+t_2}{4}}^{\frac{3t_2+t_1}{4}}\int_{x_1}^{H(t,X^n_t)}(u-f(z))\tilde{g}(t,u)dudt\right]\\
	&\geq& \ell + \frac{c}{2} E^{0,z}\left[\int_{\frac{3t_1+t_2}{4}}^{\frac{3t_2+t_1}{4}}(H(t,X^n_t)-f(z))^2dt\right],
	\eean
	where the second and third inequality are due to $\tilde{g}$ being less than or equal to $-c$ on $[t_1,t_2]\times [x_1,\infty)$. Since $H(t, X^n_t)\geq f(z)$ on $[\frac{3t_1 +t_2}{4},\frac{t_1+3t_2}{4}]$, it follows from the monotone convergence theorem that $W^n_1\rar \infty$ as $n \rar \infty$, which contradicts the definition of equilibrium. 
\end{proof}
\begin{remark}\label{mm:r:noH}
	Note that in view of the above lemma we can take $H$ to be the identity function. Indeed, if $(H,w)$ is a pricing rule satisfying the conditions of above lemma, then $(\tilde{H}, \tilde{w})$, where $\tilde{H}(t,x)=x$ and $\tilde{w}(t,x)=H_x(t,H^{-1}(t,x))w(t,H^{-1}(t,x))$ is also a pricing rule satisfying the conditions of the lemma with $\tilde{g}(t,x)=\frac{g(t,H^{-1}(t,x))}{H_x(t,H^{-1}(t,x))}$. Moreover,
	\[
	\int_{\xi(t,a)}^{x}(H(t,u)-a) g(t,u)du=\int_{a}^{H(t,x)}(u-a) \tilde{g}(t,u)du,
	\]
	and $H(t,\infty)=-H(t,-\infty)=\infty$ for all $t \in [0,1]$.
	
	Moreover, consider $S_t:=H(t,X_t)$, where $X$ is the unique solution of (\ref{mm:eq:signal_mm1}) with $c=j=0$. Then
	\[
	dS_t= \tilde{w}(t,S_{t-})dY^c_t+ \frac{\tilde{w}(t,S_{t-})\tilde{w}_x(t,S_{t-})}{2}\left(d[Y,Y]^c_t-dt\right)+(H(t,X_{t})-H(t,X_{t-}),
	\]
	and, therefore, $S$  satisfies (\ref{mm:eq:signal_mm1}) with $w=\tilde{w}$. Indeed,
	\[
	S_{t-}+H(t,X_{t})-H(t,X_{t-})=K_{\tilde{w}}^{-1}(t,K_{\tilde{w}}(t,S_{t-})+\Delta Y_t).
	\]
	The above is equivalent to 
	\[
 \Delta Y_t= K_{\tilde{w}}(t, H(t,X_t))-K_{\tilde{w}}(t,S_{t-})=\int_{S_{t-}}^{S_t}\frac{1}{\tilde{w}(t,y)}dy=\int_{X_{t-}}^{X_t}\frac{1}{w(t,y)}dy,
	\]
which holds in view of dynamics of $X$.	Therefore, without loss of generality we can assume $H$ is identity. 
\end{remark}

We can now give a proof of Theorem \ref{t:gzero}.

\begin{proof}[Proof of Theorem \ref{t:gzero}]
	Note that $\theta^*$ is continuous in view of Theorem \ref{t:thetacts}. Observe that $H^*$ can be taken equal to identity in view of Remark \ref{mm:r:noH}.  We will also restrict our attention to $z\in E$. 
	
	 Fix $z\in E$.  We show that $X^*$ must be such that  $X^*_t=\arg \min_x \int_{f(z)}^x (y-f(z)) g(t,y)dy$ for all $t\in(0,\nu)$. Suppose not, i.e. assume that there exists $t\in (0,\nu)$ and $\delta>0$ such that
	 $$
	    P^{0,z}\left[\int_{f(z)}^{X^*_t} (y-f(z)) g(t,y)dy -\int_{f(z)}^{s^{f(z)}(t)} (y-f(z)) g(t,y)dy>\delta\right]>0.
	$$
	since both $X^*$ and $s^{f(z)}$ are continuous on $[0,\nu]$ and $$\int_{f(z)}^{X^*_u} (y-f(z)) g(u,y)dy  -\int_{f(z)}^{s^{f(z)}(u)} (y-f(z)) g(u,y)dy\geq 0$$ for all $u\in[0,\nu]$ we will have 
	 $$
	    P^{0,z}\left[\int_0^{\nu}\int_{f(z)}^{X^*_t} (y-f(z)) g(t,y)dydt -\int_0^{\nu}\int_{f(z)}^{s^{f(z)}(t)} (y-f(z)) g(t,y)dydt>0\right]>0.
	$$
         This implies 
         \be\label{mm:eq:estsast}
             E^{0,z}\left[\int_0^{\nu}\int_{s^{f(z)}(t)}^{X^*_t} (y-f(z)) g(t,y)dydt\right]=\delta>0
         \ee

For any $\eps>0$ consider $s^{\eps}$ such that $s^{\eps}(0)=0$, $s^{\eps}=s^{f(z)}$ on $[\eps,\nu)$, and $s^{\eps}$ is continuous and monotone on $[0,\eps]$. Due to Lemma \ref{mm:l:appr} there exists $X^{\eps}$ such that 
\[
\sup_{t\in [0,\nu]}|s^{\eps}(t)-X^{\eps}_t|<\eps.
\]
We will use the corresponding $\theta^{\eps}$ as the insider's strategy on $[0,\nu]$.

On $[\nu, \nu+\eps]$ set $d\theta^{\eps}=-dB_t+\frac{f(z)-X^{\eps}_{\nu}}{\eps w(t,X^{\eps}_t)}dt$ and note that $X^{\eps}$ remains bounded on $[\nu,\nu+\eps]$ and $X^{\eps}_{\nu+\eps}=f(z)$. 

Now consider the interval $[\nu+\eps,\nu+2 \eps]$ and introduce
\[
dR_t^{\eps}=dB_t+ \frac{R^*_t-R^{\eps}_t}{\nu+2\eps- t}dt,
\]
with $R^{\eps}_{\nu+\eps}=K_w(\nu+\eps,f(z))$. It is easy to see that the solution to the above SDE on $[\nu+\eps,\nu+2\eps)$ is given by
\[
R_t^{\eps}=(\nu+2\eps-t)\left[R^{\eps}_{\nu+\eps}+\int_{\nu+\eps}^t\frac{1}{\nu+2\eps-s}dB_s +\int_{\nu+\eps}^t\frac{R^*_s}{(\nu+2\eps-s)^2}ds\right],
\]
where $R^*_t=K_w(t,X^*_t)$. Observe that the first two terms in the square brackets multiplies by $(\nu+2\eps-t)$ converge two $0$ as $t \rar \nu+2\eps$ in view of Exercise IX.2.12 in \cite{RY}. Moreover, on $[R^*_{\nu+2\eps}\neq 0]$ an application of L'Hospital's rule shows that the third term multiplied by $(\nu+2\eps-t)$ converges to $R^*_{\nu +2\eps}$. Similarly, on $[R^*_{\nu+2\eps}=0]$, $(\nu+2\eps-t)\int_{\nu+\eps}^t\frac{|R^*_s|}{(\nu+2\eps-s)^2}ds\rar 0$. Therefore, $R^{\eps}_{\nu+2 \eps}=R^*_{\nu+2\eps},$ a.s.. Note that if we define $\tau_R:=\inf\{t\geq \nu+\eps: R^*_t=R^{\eps}_t\}$, then $R^{\eps}$ is a semimartingale on $[\nu+\eps,\tau_R]$. Indeed,
\[
\int_{\nu+\eps}^{\tau_R}\frac{|R^*_t-R^{\eps}_t|}{(\nu+2\eps-t)}dt=\left|\int_{\nu+\eps}^{\tau_R}\frac{R^*_t-R^{\eps}_t}{(\nu+2\eps-t)}dt\right|=\left|B_{\tau_R}-R^{\eps}_{\tau_R}-B_{\nu+\eps}+R^{\eps}_{\nu+\eps}\right|<\infty.
\]

Next we define $\tilde{X_t}= K_w^{-1}(t,R^{\eps}_t)$ for $t \in [\nu+\eps,\nu+2\eps]$  and  set 
\[
X^{\eps}_t=\left\{\ba{ll}
f(z)+(\tilde{X_t}-f(z))^+, & \mbox{ if } X^*_{\nu+\eps}\geq f(z); \\
f(z)-(\tilde{X_t}-f(z))^-, & \mbox{ if } X^*_{\nu+\eps}< f(z),
\ea \right . \quad \forall t \in [\nu+\eps, \tau],
\] 
 where $\tau=\inf\{t\geq \nu+\eps: X^*_t= f(z)\}\wedge\tau_R$.  First note that $\tau \leq \nu+2\eps$ and for all $t\in [\nu+\eps,\tau]$ we have either $X^*_t\geq X^{\eps}_t\geq f(z)$ or $f(z)\geq X^{\eps}_t\geq X^*_t$ depending on  $X^*_{\nu+\eps}$. Moreover,  $X^{\eps}_{\tau}=X^*_{\tau}$ so we can set $X^{\eps}_t=X^*_t$ for $t \in [\tau, 1]$. 

$X^{\eps}$ on $[\nu+\eps,\tau]$ satisfies
\[
dX^{\eps}_t=w(t,X^{\eps}_t)(dB_t +d\theta^{\eps}_t),
\]
where in case $X^*_{\nu+\eps}\geq f(z)$
\[
d\theta^{\eps}_t=\chf_{[X^{\eps}_t >f(z)]}\left( G(t,X_t^{\eps})+\frac{R^*_t-R^{\eps}_t}{\nu+2\eps-t}\right)dt +\frac{1}{2}d\tilde{L}_t -\chf_{[X^{\eps}_t=f(z)]}dB_t,
\]
and $\tilde{L}$ is the local time of $\tilde{X}$ at $f(z)$ in view of Theorem 68 in Chap. IV of \cite{Pro}. Similarly, if $X^*_{\nu+\eps}< f(z)$,
\[
d\theta^{\eps}_t=\chf_{[X^{\eps}_t <f(z)]}\left( G(t,X_t^{\eps})+\frac{R^*_t-R^{\eps}_t}{\nu+2\eps-t}\right)dt -\frac{1}{2}d\tilde{L}_t -\chf_{[X^{\eps}_t=f(z)]}dB_t.
\]

Clearly, $\theta^{\eps}$ is admissible since $X^{\eps}$ is  bounded on $[0,\nu+\eps]$ and $|X^{\eps}_t -f(z)|\leq |X^*_t-f(z)|$ for all $t \in [\nu+\eps,1]$ and $X^*$ is admissible.

We shall show that the above strategy will outperform $\theta^*$ for small enough $\eps$. 
\bean
E^{0,z}[W^{\eps}_1-W^*_1]&\geq & E^{0,z}\left[-\int_0^{\eps}\int_{s^{f(z)}(t)}^{X^{\eps}_t} (u-f(z))g(t,u)dudt-\int_{\eps}^{\nu}\int_{s^{f(z)}(t)} ^{X^{\eps}_{t}} (u-f(z))g(t,u)dudt\right.\\
&& +\int_{\nu}^{\nu+\eps}\int_{f(z)}^{X^*_{t}} (u-f(z))(g(t,u)+C)dudt-\int_{\nu}^{\nu+\eps}\int^{X^{\eps}_t}_{f(z)} (u-f(z))g(t,u)dudt\\
&&-C\int_{\nu}^{\nu+\eps}\int_{f(z)}^{X^*_{t}} (u-f(z))dudt + \int_{\nu+\eps}^{\nu+2\eps}\int_{X^{\eps}_t}^{X^*_t} (u-f(z))(g(t,u)+C)dudt \\
&&-\left. C\int_{\nu+\eps}^{\nu+2\eps}\int_{X^{\eps}_t}^{X^*_t} (u-f(z))dudt-\eps \ell\right]+\delta\\
&\geq &-\eps \ell_1- \frac{C}{2}E^{0,z}\left[ \int_{\nu}^{\nu+2\eps} (X^*_{t}-f(z))^2dt-\int_{\nu+\eps}^{\nu+2\eps} (X^{\eps}_{t}-f(z))^2dt\right]+\delta,
\eean
where the first inequality is due to  (\ref{mm:eq:estsast}) and the boundedness of $X^{\eps}$ on $[\nu,\nu+\eps]$ in conjunction with the following estimate:
\bean
-\int_{\nu}^{1}w(t,X_t^{\eps})d[\theta^{\eps},\theta^{\eps}]_t+\int_{0}^{1}w(t,X_t^{*})d[\theta^{*},\theta^{*}]_t&=& -\int_{\nu}^{\tau}w(t,X_t^{\eps})d[\theta^{\eps},\theta^{\eps}]_t+\int_{0}^{\tau}w(t,X_t^{*})d[\theta^{*},\theta^{*}]_t\\
&\geq& -\int_{\nu}^{\nu+\eps}w(t,X_t^{\eps})dt+\int_{\nu+\eps}^{\tau}\chf_{[X^{\eps}_t=f(z)]}w(t,f(z))dt\geq \eps \ell.
\eean

  The second inequality is due to the fact that the first, second and the fourth terms are integrals of continuous functions on compact domains whose measures are proportional to $\eps$, and the sixth is positive since by construction either $X^*\geq X^{\eps}\geq f(z)$ or $X^*\leq X^{\eps}\leq f(z)$ on $[\nu+\eps, \nu+2\eps]$ and $g$ is bounded from below by some constant $-C$ in view of Lemma \ref{mm:l:gbbelow}.  Finally, the above lower estimate converges to $\delta$ as $\eps \rar 0$ due to the square integrability of $X^{\eps}$ and $X^{\eps}$.

Thus, $\theta^*$ should be such that $X^*_t= \arg \min_x \int_{f(z)}^x (y-f(z)) g(t,y)dy$ for all $t\in(0,\nu)$ for $t\in[0,\nu]$. 

Define 
$$
    E_t(s):=\{a\in\bbR: s=\arg \min
    _x \int_{a}^x (y-a) g(t,y)dy\},
$$ 
 a set of realisations of insider's signal that allow the insider to set $X^*_t=s$. 

In what follows we will show that : i) $E_t(s)$ is a connected set for all $t\in[0,\nu]$, and ii) for all $E_t(s)$ such that $s\in E_t(s)$ and $\{s\}\neq E_t(s)$ we have $g(t,y)=0$ for all $y\in E_t(s)$.

This will imply that $g\equiv 0$ for all $t\in[0,\nu)$. Indeed, suppose there exists $(t,a)\in [0,\nu)\times \bbR$ such that $g(t,a)\neq 0$. Since $g$ is continuous, there exists a set $[a_1,a_2]\subseteq \bbR$ such that $g(t,a)\neq 0$ for all $a\in [a_1,a_2]$.  Due to ii) we must have $X^*_t=a$ on the set $\{f(Z_1)=a\}$ for all  $a\in [a_1,a_2]$. Indeed, suppose on the set $\{f(Z_1)=a\}$ we have $X^*_t=s\neq a$. Since $s=X^*_t= \arg \min_x \int_{f(z)}^x (y-f(z)) g(t,y)dy$, ii) yields that $s\notin E_t(s)$ (since $a\in E_t(s)$), but then we have a contradiction to the rationality of the pricing rule as $s=X_t^*=\bbE[f(Z_1)|\cF_t^M]=\bbE[f(Z_1)\chf_{[f(Z_1)\in E_t(s)]}|\cF_t^M]\in E_t(s)$, where the last inclusion is due to i). Thus, we have $X^*_t=a$ on the set $\{f(Z_1)=a\}$ for all  $a\in [a_1,a_2]$, or, equivalently, $X^*_t=f(Z_1)$ for all $Z_1\in f^{-1}([a_1,a_2])$, but this means that the rational pricing rule must satisfy $X^*_s= X^*_t$ for any $s\geq t$ which is not consistent with the definition of the pricing rule since $w>0$.

 \begin{itemize}
   \item[i)] {\bf $\mathbf{E_t(s)}$ is connected:} Suppose $a_1,a_2\in E_t(s)$ and $\lambda\in[0,1]$. We need to show that $a=\lambda a_1+(1-\lambda)a_2\in E_t(s)$. Indeed,  for any $r\in \bbR$ we have:
           \begin{eqnarray*}
                \int_{a}^r (y-a) g(t,y)dy-\int_{a}^s (y-a) g(t,y)dy&=& \int_{s}^r (y-a) g(t,y)dy\\
                &=&\lambda\left(\int_{a_1}^r (y-a_1) g(t,y)dy-\int_{a_1}^s (y-a_1) g(t,y)dy\right)\\
                &&+(1-\lambda)\left(\int_{a_2}^r (y-a_2) g(t,y)dy-\int_{a_2}^s (y-a_2) g(t,y)dy\right)\geq 0,                
           \end{eqnarray*}
           where the last inequality is due to the fact that for $i=1,2$
           $$
               \int_{a_i}^s (y-a_i) g(t,y)dy=\min_{r\in\bbR}\int_{a_i}^r (y-a_i) g(t,y)dy.
           $$
           This implies that   $\int_{a}^r (y-a) g(t,y)dy\geq \int_{a}^s (y-a) g(t,y)dy$ for all $r\in \bbR$ and therefore $a\in E_t(s)$.
    \item[ii)] {\bf{$\mathbf{ s\in E_t(s)}$ and $\mathbf{ \{s\}\neq E_t(s)}$ $\mathbf{ \Rightarrow}$ $\mathbf{ g(t,y)=0}$ for all $\mathbf{ y\in E_t(s)}$.}} Suppose there exist $(t,s)\in[0,\nu]\times \bbR$ such that: $s\in E_t(s)$, $E_t(s)\neq \{s\}$, and $g(t,\tilde{y})\neq 0$ for some $\tilde{y}\in E_t(s)$.

            Since $s\in E_t(s)$ for any $r\in \bbR$ we will have
            $$
               \int_{s}^r (y-s) g(t,y)dy\geq \int_{s}^s (y-s) g(t,y)dy=0.
            $$
            Let $G(t,y)=\int_s^yg(t,u)du$. We have 
            \begin{eqnarray}
                0\leq \int_{s}^r (y-s) g(t,y)dy&=& \int_{s}^r (y-s) dG(t,y)=(r-s) G(t,r)-\int_{s}^r G(t,y)dy \nn\\
                                                   &=&(r-s)(G(t,r)-G(t,\psi)) \label{mm:eq:GonR}
            \end{eqnarray}
            for any $r\in \bbR$  and some $\psi\in [r\wedge s, r\vee s]$.

            Note that $\tilde{y}\neq s$ since for any $a\in E_t(s)$ we have $s=\arg \min_x \int_{a}^x (y-a) g(t,y)dy$ and therefore the first order conditions imply $g(t,s)=0$ as we can choose $a\neq s$. 
            
            Suppose $\tilde{y}>s$. Let $G(t,y^*)=\min_{y\in[s,\tilde{y}]} G(t,y)$. 
            Due to (\ref{mm:eq:GonR}) there exists $\psi\in [s, y^*]$ such that 
            \be\label{mm:eq:G_min}
              (y^*-s) G(t,y^*)-\int_{s}^{y^*} G(t,y)dy=(y^*-s)(G(t,y^*)-G(t,\psi))\geq 0
            \ee
            and therefore $G(t,y^*)\geq G(t,\psi)$. Thus, $G(t,y^*)= G(t,\psi)$ and (\ref{mm:eq:G_min}) implies            
            $$
                 \int_{s}^{y^*} G(t,y)dy=(y^*-s)G(t,\psi)=(y^*-s)G(t,y^*)=\int_{s}^{y^*} \min_{y\in[s,y^*]}G(t,y)dy.
            $$
            This yields $G(t,y)=const$ for $y\in[s,y^*]$ and in particular $G(t,s)=G(t,y^*)=\min_{y\in[s,\tilde{y}]} G(t,y)$. Since $G(t,y)=const$ for $y\in[s,y^*]$ and $g$ is continuous, we have $g(t,y)=0$ for $y\in[s,y^*]$. 
            Thus $\tilde {y}\neq y^*$ which implies $G(t,s)=\min_{y\in[s,\tilde{y}]} G(t,y)<G(t,\tilde{y})$.

             Since $\tilde{y}\in E_t(s)$ we have 
            \begin{eqnarray*}
                (s-\tilde{y})(G(t,s)-G(t,\phi))&=&(s-\tilde{y}) dG(t,s)-\int_{\tilde{y}}^s G(t,y)dy=\int_{\tilde{y}}^s (y-\tilde{y}) dG(t,y)\\
                                                           & =&\int_{\tilde{y}}^s (y-\tilde{y}) g(t,y)dy\leq \int_{\tilde{y}}^{\tilde{y}} (y-\tilde{y}) g(t,y)dy=0
            \end{eqnarray*}
            for some $\phi\in [ s, \tilde{y}]$. Thus $G(t,\phi)=G(t,s)=\min_{y\in[0,\tilde{y}]}G(t,y)$ and therefore 
            $$
             \int_s^{\tilde{y}} \min_{y\in[0,\tilde{y}]}G(t,y)dy= (\tilde{y}-s)G(t,\phi)=\int_s^{\tilde{y}} G(t,y)dy.
            $$
            Since $G$ is continuous it implies that $G(t,y)=G(t,s)$ for all $y\in[0,\tilde{y}]$ which contradicts the above result  that $G(t,s)<G(t,\tilde{y})$. Thus, we can not have $\tilde{y}\geq s$.
            
             Suppose $\tilde{y}<s$. Let $G(t,y^*)=\max_{y\in[\tilde{y},s]} G(t,y)$. 
            Due to (\ref{mm:eq:GonR}) there exists $\psi\in [ y^*,s]$ such that 
            \be\label{mm:eq:G_max}
              (y^*-s) G(t,y^*)-\int_{s}^{y^*} G(t,y)dy=(y^*-s)(G(t,y^*)-G(t,\psi))\geq 0
            \ee
            and therefore $G(t,y^*)\leq G(t,\psi)$. Thus, $G(t,y^*)= G(t,\psi)$ and (\ref{mm:eq:G_max}) implies            
            $$
                 \int_{s}^{y^*} G(t,y)dy=(y^*-s)G(t,\psi)=(y^*-s)G(t,y^*)=\int_{s}^{y^*} \max_{y\in[y^*,s]}G(t,y)dy.
            $$
            This yields $G(t,y)=const$ for $y\in[y^*,s]$ and in particular $G(t,s)=G(t,y^*)=\max_{y\in[\tilde{y},s]} G(t,y)$. Since $G(t,y)=const$ for $y\in[y^*,s]$ and $g$ is continuous, we have $g(t,y)=0$ for $y\in[y^*,s]$. 
            Thus $\tilde {y}\neq y^*$ which implies $G(t,s)=\max_{y\in[\tilde{y},s]} G(t,y)>G(t,\tilde{y})$.

             Since $\tilde{y}\in E_t(s)$ we have 
            \begin{eqnarray*}
                (s-\tilde{y})(G(t,s)-G(t,\phi))&=&(s-\tilde{y}) dG(t,s)-\int_{\tilde{y}}^s G(t,y)dy=\int_{\tilde{y}}^s (y-\tilde{y}) dG(t,y)\\
                                                           & =&\int_{\tilde{y}}^s (y-\tilde{y}) g(t,y)dy\leq \int_{\tilde{y}}^{\tilde{y}} (y-\tilde{y}) g(t,y)dy=0
            \end{eqnarray*}
            for some $\phi\in [ s, \tilde{y}]$. Thus $G(t,\phi)=G(t,s)=\max_{y\in[\tilde{y},s]}G(t,y)$ and therefore 
            $$
             \int_s^{\tilde{y}} \min_{y\in[0,\tilde{y}]}G(t,y)dy= (\tilde{y}-s)G(t,\phi)=\int_s^{\tilde{y}} G(t,y)dy.
            $$
            Since $G$ is continuous it implies that $G(t,y)=G(t,s)$ for all $y\in[0,\tilde{y}]$ which contradicts the above result  that $G(t,s)<G(t,\tilde{y})$. 
            Thus, there doesn't exist $\tilde{y}\in E_t(s)$ such that $g(t,\tilde{y})\neq 0$.
 \end{itemize}
\end{proof}
Finally, we can prove Theorem \ref{mm:t:AC}.
\begin{proof}[Proof of Theorem \ref{mm:t:AC}]
In view of (\ref{mm:eq:uljumps}) and since $w$ is positive  and $H$ is increasing, we have
	\[
	E^{0,z}\left[W_1^{\theta}\right]\leq E^{0,z}\left[\Psi^{f(Z_1)}(0,0)-\Psi^{f(Z_1)}(1-,X_{1-})\right]
	\]
	Note the inequality above becomes equality if and only if $\Delta \theta_t=0$ due to the strict monotonicity of $H$. Moreover, $\Psi^{f(Z_1)}(1-,X_{1-})\geq 0$ with an equality if and only if $H(1-, X_{1-})=f(Z_1)$. Therefore, $E^{0,z}\left[W_1^{\theta}\right]\leq E^{0,z}\left[\Psi^{f(Z_1)}(0,0)\right]$ for all admissible $\theta$s and equality is reached if and only if the following two conditions are met: 
	\begin{itemize}
		\item[i)] $\theta$ is continuous and of finite variation.
		\item[ii)] $H(1-,X_{1-})=f(Z_1), \, P^{0,z}$-a.s..
	\end{itemize} 
	
	Hence, the proof will be complete if one can find a sequence of  absolutely continuous admissible strategies, $(\theta^n)_{n \geq 1}$ such that $\lim_{n \rar \infty}E^{0,z}\left[W_1^{\theta^n}\right]=E^{0,z}\left[\Psi^{f(Z_1)}(0,0)\right]$. 
	
 Consider the bridge process, $Y$, that starts at $0$ and ends up at $M_1$ at $t=1$:
	\[
	Y_t:=B_t +\int_0^t \frac{M_s-Y_s}{1-s}ds=(1-t)\left(\int_{0}^t\frac{1}{1-s}dB_s +\int_0^t \frac{M_s}{(1-s)^2}ds\right).
	\]
It is easy to check that the above converges a.s. to $M_1$ using the continuity of $M$ and L'Hospital rule since $(1-t)\int_{0}^t\frac{1}{1-s}dB_s$ is the Brownian bridge from $0$ to $0$ as in Exercise IX.2.12 in \cite{RY}.

	To establish the semimartingale property of $Y$ first observe that
	\[
	Y_t-M_t = (1-t)\int_0^t \frac{1}{1-s}\{dB_s-dM_s\}.
	\]
	Thus, by Theorem V.1.6 in \cite{RY}, there exists a Brownian motion $\tilde{\beta}$ such that
	\[
	\int_0^1\frac{|M_t-Y_t|}{1-t}dt=\int_0^1 |\tilde{\beta}_{\tau_t}|dt,
	\]
	where 
	\[
	\tau_t=\int_0^t\frac{1+\tilde{\sigma}^2_s}{(1-s)^2}ds.
	\]
	Observe that $\tau_{1}=\infty$, $P^{0,z}$-a.s.. Thus by the law of iterated logarithm for Brownian motion (see Corollary II.1.12 in \cite{RY}), we have
	\[
	\frac{|\tilde{\beta}_{\tau_t}|}{\sqrt{ \tau_t \log \log \tau_t}}<C \; \forall t \in [0,1], \, P^{0,z}\mbox{-a.s.}
		\]
 for some finite random variable $C$. 	Therefore,
 \be \label{e:Ysemibd}
 	\int_0^1\frac{|M_t-Y_t|}{1-t}dt<C \int_0^1 \sqrt{\tau_t \log \log \tau_t}dt\leq C\int_0^1\tau_t^{\frac{1+\eps}{2}}dt,
 \ee
 for all $\eps>0$. Thus, $Y$ is a semimartingale.
 
 Note that (\ref{e:MQVbd}) implies that for any $n>1$ there exists $\delta>0$ such that for any $s \in [1-\delta,1]$,  $\tilde{\sigma}^2_s (1-s)^{\alpha-1}<\frac{1}{n}$. Therefore, for $t \geq 1-\delta$,
 \[
 (1-t)^{\alpha}\int_0^t \frac{1+\tilde{\sigma}^2_s}{(1-s)^2}ds\leq (1-t)^{\alpha-1}-(1-t)^{\alpha} + (1-t)^{\alpha}\left(\int_0^{1-\delta} \frac{1+\tilde{\sigma}^2_s}{(1-s)^2}ds+ \frac{1}{n}\left((1-t)^{-\alpha}-\delta^{-\alpha}\right)\right),
 \]
 which in turn yields $\lim_{t \rar 1}\tau_t (1-t)^{\alpha}=0$. Hence,
 \[
 \int_0^1\tau_t^{\frac{1+\eps}{2}}dt\leq \tilde{C}\int_0^1 (1-t)^{-\frac{\alpha(1+\eps)}{2}}dt<\infty
 \]
for any $\eps <\frac{2}{\alpha} -1$. This proves the semimartingale property of $Y$ in view of  (\ref{e:Ysemibd}).

	Next define the stopping times 
	\[
	\tau^{n}:=\inf\{t:|Y_t|\geq n\}
	\]
	with the convention that $\inf \emptyset =1$, and introduce the sequence of trading strategies, $\theta^{n}$ given by
	\[
	d\theta^{n}_t=\chf_{[\tau^{n}\geq t]}\frac{M_t-Y_t}{1-t}dt + \chf_{[\tau^{n,m}< t]}d\tilde{\theta}^n_t,
	\]
	where $\tilde{\theta}^n$ is the continuous and of finite variation process given by Lemma \ref{mm:l:appr} to keep $Y^n_t \in (-1-n,1+n)$ for $t\geq\tau^n$ via choosing $x(t)=Y_{\tau^n}\frac{1-t}{1-\tau^n}$ and $\eps=w=1$.  This will also ensure that $Y^n_1 \in (-1,1)$ on $[\tau^n<1]$.
	Thus, the total demand process $Y^n$ corresponding to $\theta^{n}$ satisfies 
	\begin{enumerate}
		\item $\sup_{t\in [0,1]}|Y^{n}_t|\leq n+1$, a.s.;
		\item $Y^{n}_1\chf_{[\tau^{n}=1]}=Y_1\chf_{[\tau^{n}=1]}=K_w(1,H^{-1}(1,f(Z_1)))\chf_{[\tau^{n}=1]}$, a.s..
		\item $Y^{n}_1\chf_{[\tau^{n}<1]} \in (-1,1)$.
	\end{enumerate}
	In view of Remark \ref{mm:r:Kwubdd} and the continuity of $H(1, K_w^{-1}(1,\cdot))$ we deduce that  $H(t, K_w^{-1}(t,Y^{n}_t))$ is bounded uniformly in $t$ yielding $\theta^{n}$ admissible for each $n$. 
	
	Recall that since $\theta^{n}$ is absolutely continuous, we have
	\[
	E^{0,z}[W^{\theta^{n}}]=E^{0,z}\left[\Psi^{f(Z_1)}(0,0)-\Psi^{f(Z_1)}(1,K_w^{-1}(1,Y^{n}_{1}))\right].
	\]
	On the other hand,
	\bean
	\Psi^{f(Z_1)}(1,K_w^{-1}(1,Y^{n}_{1})) &\leq& (H(1,K_w^{-1}(1,Y^{n}_1))-f(Z_1))(Y^{n}_1-K_w(1,H^{-1}(1,f(Z_1))))\\
	&=&\chf_{[\tau^{n}<1]}(H(1,K_w^{-1}(1,Y^{n}_1))-f(Z_1))(Y^{n}_1-K_w(1,H^{-1}(1,f(Z_1)))).
	\eean
	Since $f(Z_1)K_w(1,H^{-1}(1,f(Z_1)))$ is integrable and $Y^{n}_1$ is uniformly bounded, applying the dominated convergence theorem  yields 
	\[
	\lim_{n \rar \infty}E^{0,z}[W^{\theta^{n}}]=E^{0,z}\left[\Psi^{f(Z_1)}(0,0)\right],
	\]
	i.e. the expected wealth corresponding to our sequence of admissible strategies converges to the upper limit of the value function.
\end{proof}
\appendix
\section{Auxiliary results}
\begin{remark}\label{mm:r:Kwubdd} Observe that the existence of a unique strong solution in Definition \ref{mm:d:prule} implies $\min\{\bbP(K_w(t,\tilde{X}_t)>y), \bbP(K_w(t,\tilde{X}_t)<-y)\}>0$. In particular  $K_w(t,\cdot):\bbR \to\bbR$  is onto for every $t \in [0,1]$. 
	
	Indeed,
	$\bbP(K_w(t,\tilde{X}_t)>y)\geq \bbP(K_w(t,\tilde{X}_t)>y, \sup_{s\leq t}|\tilde{X}_s|\leq n)$ for some large enough $n$. On the other hand,
	application of Ito's formula yields
	\[
	K_w(t,\tilde{X}_t)= \beta_t -\int_0^t G(s,\tilde{X}_s)ds,
	\]
	where $G(t,x)=\int_0^x g(t,y)dy$ and $g(t,x):=\frac{w_t(t,x)}{w^2(t,x)}+\half w_{xx}(t,x)$ are continuous. 
	
	Thus, the law of $K_w(t\wedge \tau_n,\tilde{X}_{t\wedge \tau_n})_{t\in [0,1]}$ is equivalent to that of $(\beta_{t \wedge \nu_n})_{t \in [0,1]}$, where $\tau_n=\inf\{t\geq 0:|\tilde{X}_t|\geq n\}$ and $\nu_n=\{t\geq 0: |K^{-1}_w(t,\beta_t)|\geq n\}$ by Girsanov's theorem. Therefore, $\bbP(K_w(t,\tilde{X}_t)>y, \sup_{s\leq t}|\tilde{X}_s|\leq n)>0$ yielding the claim. Similarly, $\bbP(K_w(t,\tilde{X}_t)<-y, \sup_{s\leq t}|\tilde{X}_s|\leq n)>0$. Consequently, $\min\{\bbP(\tilde{X}_t>x), \bbP(\tilde{X}_t<-x)\}>0$ for all $t \in [0,1]$ by choosing $y=K_w(t,x)$.
\end{remark}
\begin{lemma} \label{mm:l:appr}
	Consider bounded stopping times $S \leq T$ and let $x:[S,T] \mapsto \bbR$  be continuous, adapted, and of finite variation. Then for any  $\eps>0$ there exists an adapted process $\theta^{\eps}$ that is continuous and of finite variation on $[S,T]$ such that there exists a strong solution to 
	\[
	dX^{\eps}_t=w(t,X_t)\{dB_t +d\theta^{\eps}_t\}
	\]
	with $X^{\eps}_S=x(S)$ for a given $w:[0,1]\times \bbR\to (0,\infty) \in C^{1,2}$ . Moreover, $X^{\eps}$ satisfies
	\[
	\sup_{r \in [S,T]}|X^{\eps}_r-x(r)|<\eps.
	\]
\end{lemma}
\begin{proof}
	Define $y(t):=K_{w}(t,x(t))$ and observe that $y$ is continuous and of finite variation. Moreover, introduce the stochastic process $U^{\delta}$ with $U^{\delta}_S=0$ and
	\[
	dU^{\delta}_t= dB_t + \left(\frac{1}{U^{\delta}_t+\delta}\chf_{[U^{\delta}_t\leq 0]}-\frac{1}{\delta-U^{\delta}_t}\chf_{[U_t^{\delta}>0]}\right)dt,
	\]	
	which stays in $(-\delta,\delta)$ in view of Proposition 3.1 in \cite{rtr}. 
	
	Next define $R^{\delta}:=U^{\delta}+y$ on $[S,T]$ and set $X^{\delta}_t=K^{-1}_{w}(t,R^{\delta}_t)$. Thus,
	\[
	dX^{\delta}_t= w(t,X^{\delta}_t) \left\{dB_t +d\theta^{\delta}_t\right\},
	\]
	where
	\[
	d\theta^{\delta}_t=\left(G(t,K_w^{-1}(t,R_t^{\delta}))+\frac{1}{U^{\delta}_t+\delta}\chf_{[U^{\delta}_t\leq 0]}-\frac{1}{\delta-U^{\delta}_t}\chf_{[U_t^{\delta}>0]}\right)dt+dy_t,
	\]
	where $G(t,x):=\int_0^{x}g(t,y)dy$.
	Therefore,
	\be \label{mm:eq:seps-xeps}
	\sup_{t \in [S,T]}|y(t)-K_w(t,X^{\delta}_t)|<\delta.
	\ee
	Choosing $\delta$ small enough we thus obtain
	\[
	\sup_{t \in [S,T]}|x(t)-X^{\delta}_t|=\sup_{t \in [S,T]}|K_w^{-1}(t,y(t))-X^{\delta}_t)|<\eps.
	\]
	due to the uniform continuity of $K^{-1}_w$ on compacts. 
\end{proof}
\begin{lemma}
	Let $g:[0,1]\times \bar{\bbR}\to \bar{\bbR}$ be continuous such that $g:[0,1]\times \bbR\to \bbR$ is also continuous. Suppose that $\lim_{u \rar \infty}g(t,u) \geq 0$ and $\lim_{u \rar -\infty}g(t,u) \geq 0$ for every $t \in [0,1]$. Then $g$ is bounded from below.
\end{lemma}
\begin{proof} Consider sets $(E_n)_{n \geq 1}$
	\[
	E_n:=\{(t,u)\in [0,1]\times \bbR:g(t,u)\leq -n \}.
	\]
	Clearly, $A_n$s are closed. They are also bounded. Indeed, if there exists a sequence $(t_m,u_m)\in A_n$ such that $u_m \rar \infty$ or $u_m \rar -\infty$. Then, $\lim_{m\rar \infty}g(t_m,u_m) \leq -n$, contradicting the hypotheses on joint continuity and the limits at $\pm \infty$. Thus, $A_n$s are compact. Therefore, if all $A_n$s are non-empty, then $\cap_{n \geq 1} A_n \neq \emptyset$ by the nested set property (see Corollary to Theorem 2.36 in \cite{RudinMA}). By construction $g(t,u)=-\infty$ for any  $(t,u)$ in this intersection, which contradicts our continuity assumption on $g$, Therefore, $A_n$ must be empty for all $n>N$ for some $N$.
\end{proof}
\begin{proof}[Proof of Theorem \ref{t:thetacts}]
	First observe that $H^*$ can be taken equal to identity in view of Remark \ref{mm:r:noH}. Let $\nu:=\inf\{t\geq 0: \Delta X_t>0\}$ and suppose $P^{0,z}(\nu<1)>0$.  We will construct a strategy $\theta^{\eps}$ that agrees with $\theta^*$ on $[0,\nu)$.  Let $\eps>0$ and choose $\eps(\nu)=\eps \wedge \frac{1-\nu}{3}$. 
	
	On $[\nu, \nu+\eps(\nu)]$ set $X_{\nu}^{\eps}=X^*_{\nu-}, \, d\theta^{\eps}=-dB_t+\frac{f(z)-X^{\eps}_{\nu}}{\eps(\nu) w(t,X^{\eps}_t)}dt$ and note that $|X^{\eps}|\leq |X^*_{\nu-}|+|f(z)|$  on $[\nu,\nu+\eps(\nu)]$ as well as $X^{\eps}_{\nu+\eps(\nu)}=f(z)$.
	
	Now consider the interval $[\nu+\eps(\nu),\nu+2 \eps(\nu)]$ and introduce
	\[
	dR_t^{\eps}=dB_t+ \frac{R^*_t-R^{\eps}_t}{\nu+2\eps(\nu)- t}dt,
	\]
	with $R^{\eps}_{\nu+\eps(\nu)}=K_w(\nu+\eps(\nu),f(z))$. It is easy to see that the solution to the above SDE on $[\nu+\eps(\nu),\nu+2\eps(\nu))$ is given by
	\[
	R_t^{\eps}=(\nu+2\eps(\nu)-t)\left[R^{\eps}_{\nu+\eps(\nu)}+\int_{\nu+\eps(\nu)}^t\frac{1}{\nu+2\eps(\nu)-s}dB_s +\int_{\nu+\eps(\nu)}^t\frac{R^*_s}{(\nu+2\eps(\nu)-s)^2}ds\right],
	\]
	by $(\nu+2\eps(\nu)-t)$ converge two $0$ as $t \rar \nu+2\eps(\nu)$ in view of Exercise IX.2.12 in \cite{RY}. Moreover, on $[R^*_{\nu+2\eps(\nu)}\neq 0]$ an application of L'Hospital's rule shows that the third term multiplied by $(\nu+2\eps(\nu)-t)$ converges to $R^*_{(\nu +2\eps(\nu))-}$. Similarly, on $[R^*_{(\nu+2\eps(\nu))-}=0]$, $(\nu+2\eps(\nu)-t)\int_{\nu+\eps(\nu)}^t\frac{|R^*_s|}{(\nu+2\eps(\nu)-s)^2}ds\rar 0$. Therefore, $R^{\eps}_{\nu+2 \eps(\nu)}=R^*_{(\nu+2\eps(\nu))-},$ a.s.. Note that if we define 
	\[
	\tau_R:=\inf\{t\geq \nu+\eps(\nu): \sgn(R^*_t-R^{\eps}_t)\neq \sgn(R^*_{\nu+\eps(\nu)}-R^{\eps}_{\nu+\eps(\nu)})\}\wedge \inf\{t\geq \nu+\eps(\nu): R^*_t=R^{\eps}_t\},
	\]
	where $\sgn(x)=\chf_{x>0}-\chf_{x\leq 0}$, then $R^{\eps}$ is a semimartingale on $[\nu+\eps(\nu),\tau_R]$. Indeed,
	\[
	\int_{\nu+\eps(\nu)}^{\tau_R}\frac{|R^*_t-R^{\eps}_t|}{(\nu+2\eps(\nu)-t)}dt=\left|\int_{\nu+\eps(\nu)}^{\tau_R}\frac{R^*_t-R^{\eps}_t}{(\nu+2\eps(\nu)-t)}dt\right|=\left|B_{\tau_R}-R^{\eps}_{\tau_R}-B_{\nu+\eps(\nu)}+R^{\eps}_{\nu+\eps(\nu)}\right|<\infty.
	\]
	
	Next we define $\tilde{X_t}= K_w^{-1}(t,R^{\eps}_t)$ for $t \in [\nu+\eps(\nu),\nu+2\eps(\nu))$  and  set 
	\[
	X^{\eps}_t=\left\{\ba{ll}
	f(z)+(\tilde{X_t}-f(z))^+, & \mbox{ if } X^*_{\nu+\eps(\nu)}\geq f(z); \\
	f(z)-(\tilde{X_t}-f(z))^-, & \mbox{ if } X^*_{\nu+\eps(\nu)}< f(z),
	\ea \right . \quad \forall t \in [\nu+\eps(\nu), \tau),
	\] 
	where 
	\[
	\tau=\inf\{t\geq \nu+\eps(\nu): X^*_t= f(z)\}\wedge \inf\{t\geq \nu+\eps(\nu): \sgn(X^*_t- f(z))\neq \sgn(X^*_{\nu+\eps(\nu)}- f(z))\}\wedge\tau_R.
	\]
	
	$X^{\eps}$ on $[\nu+\eps(\nu),\tau)$ satisfies
	\[
	dX^{\eps}_t=w(t,X^{\eps}_t)(dB_t +d\theta^{\eps}_t),
	\]
	where in case $X^*_{\nu+\eps(\nu)}\geq f(z)$
	\[
	d\theta^{\eps}_t=\chf_{[X^{\eps}_t >f(z)]}\left( G(t,X_t^{\eps})+\frac{R^*_t-R^{\eps}_t}{\nu+2\eps(\nu)-t}\right)dt +\frac{1}{2}d\tilde{L}_t -\chf_{[X^{\eps}_t=f(z)]}dB_t,
	\]
	and $\tilde{L}$ is the local time of $\tilde{X}$ at $f(z)$ in view of Theorem 68 in Chap. IV of \cite{Pro}. Similarly, if $X^*_{\nu+\eps(\nu)}< f(z)$,
	\[
	d\theta^{\eps}_t=\chf_{[X^{\eps}_t <f(z)]}\left( G(t,X_t^{\eps})+\frac{R^*_t-R^{\eps}_t}{\nu+2\eps(\nu)-t}\right)dt -\frac{1}{2}d\tilde{L}_t -\chf_{[X^{\eps}_t=f(z)]}dB_t.
	\]
	Next pick an $n\geq 1$ and consider $\hat{\theta}$, which is given by
	\[
	\hat{\theta}_t=\chf_{[t<\nu]}\theta^*_t+ \chf_{[t\geq \nu]}\left(\theta^*_t\chf_{[X^*_{\nu-}>n]}+\theta^{\eps}_t\chf_{[X^*_{\nu-}\leq n]}\right).
	\]
	This strategy is clearly admissible and  will outperform $\theta^*$ for small enough $\eps$ and large enough $n$ by following the reasoning and calculations that led to the analogous conclusion in Theorem \ref{t:gzero}. 
\end{proof}
\bibliographystyle{siam}
\bibliography{ref}

\begin{thebibliography}{10}

\bibitem{Back}
{\sc K.~Back}, {\em Insider trading in continuous time}, Review of Financial
  Studies, 5 (1992), pp.~387--409.

\bibitem{Back-Baruch}
{\sc K.~Back and S.~Baruch}, {\em Information in securities markets: {K}yle
  meets {G}losten and {M}ilgrom}, Econometrica, 72 (2004), pp.~433--465.

\bibitem{BP}
{\sc K.~Back and H.~Pedersen}, {\em Long-lived information and intraday
  patterns}, Journal of Financial Markets,  (1998), pp.~385--402.

\bibitem{Baruch}
{\sc S.~Baruch}, {\em Insider trading and risk aversion}, Journal of Financial
  Markets, 5 (2002), pp.~451 -- 464.

\bibitem{Bauer}
{\sc H.~Bauer}, {\em Probability theory}, vol.~23 of De Gruyter Studies in
  Mathematics, Walter de Gruyter \& Co., Berlin, 1996.
\newblock Translated from the fourth (1991) German edition by Robert B. Burckel
  and revised by the author.

\bibitem{CCD}
{\sc L.~Campi, U.~{\c{C}}etin, and A.~Danilova}, {\em Dynamic {M}arkov bridges
  motivated by models of insider trading}, Stochastic Processes and their
  Applications, 121 (2011), pp.~534--567.

\bibitem{CCD2}
{\sc L.~Campi, U.~{\c{C}}etin, and A.~Danilova}, {\em Equilibrium model with
  default and dynamic insider information}, Finance and Stochastics, 17 (2013),
  pp.~565--585.

\bibitem{rtr}
{\sc U.~\c{C}etin}, {\em Diffusion transformations, {B}lack-{S}choles equation
  and optimal stopping}, Ann. Appl. Probab., 28 (2018), pp.~3102--3151.

\bibitem{CRH}
{\sc U.~{\c{C}}etin}, {\em Financial equilibrium with asymmetric information
  and random horizon}, Finance and Stochastics, 22 (2018), pp.~97--126.

\bibitem{CDBook}
{\sc U.~{\c{C}}etin and A.~Danilova}, {\em Dynamic Markov Bridges and Market
  Microstructure: Theory and Applications}, vol.~90 of Probability Theory and
  Stochastic Processes, Springer, 2018.

\bibitem{Cho}
{\sc K.-H. Cho}, {\em Continuous auctions and insider trading: uniqueness and
  risk aversion}, Finance Stoch., 7 (2003), pp.~47--71.

\bibitem{CDF}
{\sc P.~Collin-Dufresne and V.~Fos}, {\em Insider trading, stochastic
  liquidity, and equilibrium prices}, Econometrica, 84 (2016), pp.~1441--1475.

\bibitem{CNF}
{\sc J.~M. Corcuera, G.~Di~Nunno, and J.~Fajardo}, {\em Kyle equilibrium under
  random price pressure}, Decisions in Economics and Finance, 42 (2019),
  pp.~77--101.

\bibitem{CFNO}
{\sc J.~M. Corcuera, G.~Farkas, G.~Di~Nunno, and B.~{\O}ksendal}, {\em
  {Kyle-Back's} model with {L\'evy} noise}, preprint series in pure
  mathematics, Mathematical Institute, University of Oslo, Norway, 2010.

\bibitem{Danilova}
{\sc A.~Danilova}, {\em Stock market insider trading in continuous time with
  imperfect dynamic information}, Stochastics, 82 (2010), pp.~111--131.

\bibitem{JS}
{\sc J.~Jacod and A.~Shiryaev}, {\em Limit theorems for stochastic processes},
  vol.~288 of Grundlehren der Mathematischen Wissenschaften [Fundamental
  Principles of Mathematical Sciences], Springer-Verlag, Berlin, second~ed.,
  2003.

\bibitem{Kurtz07}
{\sc T.~Kurtz}, {\em The {Y}amada-{W}atanabe-{E}ngelbert theorem for general
  stochastic equations and inequalities}, Electron. J. Probab., 12 (2007),
  pp.~951--965.

\bibitem{Kyle}
{\sc A.~Kyle}, {\em Continuous auctions and insider trading}, Econometrica, 53
  (1985), pp.~1315--1335.

\bibitem{MSZ}
{\sc J.~Ma, R.~Sun, and Y.~Zhou}, {\em Kyle--back equilibrium models and linear
  conditional mean-field sdes}, SIAM Journal on Control and Optimization, 56
  (2018), pp.~1154--1180.

\bibitem{Pro}
{\sc P.~E. Protter}, {\em Stochastic integration and differential equations},
  vol.~21 of Stochastic Modelling and Applied Probability, Springer-Verlag,
  Berlin, 2005.
\newblock Second edition. Version 2.1, Corrected third printing.

\bibitem{RY}
{\sc D.~Revuz and M.~Yor}, {\em Continuous martingales and {B}rownian motion},
  vol.~293 of Grundlehren der Mathematischen Wissenschaften [Fundamental
  Principles of Mathematical Sciences], Springer-Verlag, Berlin, third~ed.,
  1999.

\bibitem{RudinMA}
{\sc W.~Rudin}, {\em Principles of mathematical analysis}, McGraw-Hill Book
  Co., New York-Auckland-D\"{u}sseldorf, third~ed., 1976.
\newblock International Series in Pure and Applied Mathematics.

\bibitem{GTM}
{\sc M.~Sharpe}, {\em General theory of {M}arkov processes}, vol.~133 of Pure
  and Applied Mathematics, Academic Press, Inc., Boston, MA, 1988.

\bibitem{Wu}
{\sc C.-T. Wu}, {\em Construction of Brownian motions in enlarged filtrations
  and their role in mathematical models of insider trading}, PhD thesis,
  Humboldt-Universit{\"a}t zu Berlin, Mathematisch-Naturwissenschaftliche
  Fakult{\"a}t II, 1999.

\end{thebibliography}
\end{document}